\documentclass[11pt]{article}
\usepackage[pdftex]{graphicx}
\usepackage{amssymb,amsmath,latexsym,amsthm,amsfonts}
\usepackage{verbatim}
\usepackage{lmodern}
\usepackage{verbatim}

\usepackage[T1]{fontenc}
\usepackage[utf8]{inputenc}
\usepackage{cite}
\usepackage{graphicx}
\usepackage{ifpdf} %
\usepackage{url} %
\usepackage{cancel}
\usepackage[numbers]{natbib}
\usepackage[margin=1in,bottom=0.9in,includefoot]{geometry}
\usepackage{etex,etoolbox}

\usepackage{nameref}
\usepackage{color}
\definecolor{rltred}{rgb}{0.75,0,0}
\definecolor{rltgreen}{rgb}{0,0.5,0}
\definecolor{rltblue}{rgb}{0,0,0.75}

\usepackage{thumbpdf}
\usepackage[pdftex,colorlinks=true,urlcolor=rltred, filecolor=rltgreen, linkcolor=rltblue]{hyperref}

\usepackage{xcolor}
\newcommand{\rev}[1]{#1}

\newcommand{\dis}{\displaystyle}

\usepackage{etex,etoolbox}
\usepackage{amsthm,amssymb}
\usepackage{thmtools}
\usepackage{environ}

\usepackage{hyperref}
\hypersetup{
     colorlinks   = true,
     citecolor    = gray
}

\makeatletter
\providecommand{\@fourthoffour}[4]{#4}
\newcommand\fixstatement[2][\proofname\space of]{%
  \ifcsname thmt@original@#2\endcsname
    \AtEndEnvironment{#2}{%
      \xdef\pat@label{\expandafter\expandafter\expandafter
        \@fourthoffour\csname thmt@original@#2\endcsname\space\@currentlabel}%
      \xdef\pat@proofof{\@nameuse{pat@proofof@#2}}%
    }%
  \else
    \AtEndEnvironment{#2}{%
      \xdef\pat@label{\expandafter\expandafter\expandafter
        \@fourthoffour\csname #1\endcsname\space\@currentlabel}%
      \xdef\pat@proofof{\@nameuse{pat@proofof@#2}}%
    }%
  \fi
  \@namedef{pat@proofof@#2}{#1}%
}

\globtoksblk\prooftoks{1000}
\newcounter{proofcount}

\NewEnviron{proofatend}{%
  \edef\next{%
    \noexpand\begin{proof}[\pat@proofof\space\pat@label]%
    \unexpanded\expandafter{\BODY}}%
  \global\toks\numexpr\prooftoks+\value{proofcount}\relax=\expandafter{\next\end{proof}}
  \stepcounter{proofcount}}

\def\printproofs{%
  \count@=\z@
  \loop
    \the\toks\numexpr\prooftoks+\count@\relax
    \ifnum\count@<\value{proofcount}%
    \advance\count@\@ne
  \repeat}
\makeatother

\declaretheorem[style=plain,name=Theorem,numberwithin=section]{theorem}

\newtheorem{definition}[theorem]{Definition}
\newtheorem{corollary}[theorem]{Corollary}
\newtheorem{lemma}[theorem]{Lemma}
\newtheorem{proposition}[theorem]{Proposition}
\newtheorem{example}[theorem]{Example}

\newtheorem*{theorem*}{Theorem}

\theoremstyle{remark}
\newtheorem{remark}[theorem]{Remark}

\usepackage{textcomp}

\newcommand{\?}{\' }
\newcommand{\R}{\mathbb{R}}   
\newcommand{\N}{\mathbb{N}}   

\newcommand{\EE}{\mathbb{E}}

\newcommand{\Fk}
{\mathcal{F}^k}
\newcommand{\oq}{\overline{q}}
\newcommand{\phio}{\overline{\phi}}

\newcommand{\musiclab}{\textsc{MusicLab}}

\begin{document}
\color{black}
\title{Popularity Signals in Trial-Offer Markets \\ with Social Influence and Position Bias}
\author{
Felipe Maldonado\footnote{The Australian National University and DATA61, Canberra, Australia (felipe.maldonado@data61.csiro.au)}
\and
Pascal Van Hentenryck\footnote{Industrial and Operations Engineering, University of Michigan, USA (pvanhent@umich.edu).}
\and
Gerardo Berbeglia\footnote{Melbourne Business School, University of Melbourne, Australia (g.berbeglia@mbs.edu).{\bf Corresponding author.}}
\and
Franco Berbeglia\footnote{Tepper Business School, Carnegie Mellon University, USA (fberbegl@andrew.cmu.edu)}
}
\date{\vspace{-5ex}}

\maketitle

\begin{abstract}
\noindent
This paper considers trial-offer markets where consumer preferences
are modeled by a multinomial logit with social influence and position
bias. The social signal for a product is given by its current market
share raised to power $r$ (\rev{or equivalently the number of purchases raised to the power of $r$}). The paper shows that, when $r$ is strictly
between 0 and 1, and a static position assignment (e.g., a quality
ranking) is used, the market converges to a unique equilibrium where
the market shares depend only on product quality, not their initial
appeals or the early dynamics. When $r$ is greater than 1, the market
becomes unpredictable. In many cases, the market goes to a monopoly
for some product: Which product becomes a monopoly depends on the
initial conditions of the market. These theoretical results are
complemented by an agent-based simulation which indicates that
convergence is fast when $r$ is between 0 and 1, and that the quality
ranking dominates the well-known popularity ranking in terms of market
efficiency. These results shed a new light on the role of social
influence which is often blamed for unpredictability, inequalities,
and inefficiencies in markets. In contrast, this paper shows that,
with a proper social signal and position assignment for the products,
the market becomes predictable, and inequalities and inefficiencies
can be controlled appropriately.
\end{abstract}
\smallskip
\noindent \textbf{Keywords.} System dynamics, social influence, stochastic dynamics, Robbins-Monro algorithms,  popularity signals, ranking policies.

\section{Introduction}
\label{section:introduction}

The impact of social influence and product visibilities on consumer
\rev{behaviour} in Trial-Offer (T-O) markets\footnote{A trial-offer market is
  a setting where consumers can try products before deciding whether
  to buy them or not.} has been explored in a variety of settings (e.g.,
\citep{salganik2006experimental,tucker2011does,viglia2014please}).
Social influence can be dispensed through different types of social
signals: A market place may report the number of past purchases of a
product, its consumer ratings, and/or its consumer
recommendations. Recent studies
\citep{engstrom2014demand,viglia2014please} however came to the
conclusion that the popularity signal (i.e., the number of past
purchases or the market share) has a much stronger impact on consumer
behaviour than the average consumer rating signal.\footnote{The music
  market {\sc iTunes} shows the normalised market share of each song
  of an album.} These two experimental studies were conducted in very
different settings, using the Android application platform in one case
and hotel selection in the other. Consumer preferences are also
influenced by product visibilities, a phenomenon that has been widely
observed in internet advertisement (e.g.,
\citep{craswell2008experimental}), in online stores such as {\sc
  Expedia}, {\sc Amazon}, and {\sc iTunes}, as well as physical retail
stores (see, e.g., \citep{lim2004metaheuristics}).

Despite its widespread use in online settings (including for songs,
albums, movies, hotels, and cell phones to name only a few), there is
considerable debate in the scientific community about the benefits of
social influence. Many researchers have pointed out the potential
negative effects of social influence. The seminal work of Salganik et
al.  \citep{salganik2006experimental} on the \musiclab{} experimental
market demonstrated that social influence can introduce significant
unpredictability, inequality, and inefficiency in T-O markets. These
results were reproduced by many researchers (e.g.,
~\cite{salganik2009web,salganik2008leading,Muchnik2013,vandeRijt2014}).
More recently, Hu et al. \citep{hu2015liking} studied a newsvendor
problem with two substitutable products with the same quality in which
consumer preferences are affected by past purchases. The authors
showed that the market is unpredictable but can become less so if one
of products has an initial advantage. Altszyler et al. \citep{altszyler2017transient} has recently studied the impact of product appeal
 and product quality in a trial-offer market model with social influence under a finite time horizon. The authors showed that there exists a logarithmic tradeoff between the two: the final product market share remains constant if a decrease in product quality 
 is followed by an exponential increase in the product appeal. Other researchers have focused
on understanding when these undesirable side-effects arise and where
they come from. Ceyhan et al. \citep{Ceyhan11} studied a market
specified by a logit model where a constant $J$ captures the strength
of the social signal. They showed that the market \rev{behaviour} (e.g.,
whether it is predictable) depends on the strength of the social
signal. Their results did not consider product visibilities, which is
another important aspect of T-O markets. Indeed, various researchers
(e.g., \cite{lerman2014leveraging,Rank,ICWSM16SI,IJCAI16SI,abeliuk2017taming}) indicated
that unpredictability and inefficiencies often depend on how products
are displayed in the market. In particular, Abeliuk et al.
\cite{Rank} show that social influence is always beneficial in
expectation when the products are ordered by the {\em performance
  ranking} that maximises the purchases greedily at each step. This
result was obtained using the generalised multinomial logit model
proposed by Krumme et al. \citep{krumme2012quantifying} to reproduce
the \musiclab{} experiments. Van Hentenryck et al. \cite{ICWSM16SI}
proves a similar result for the {\em quality ranking} that assigns the
highest quality products to the most visible positions. In addition,
they show that the market converges to a monopoly for the highest
quality product.  These results contrast with the \musiclab{}
experiments which relied on the {\em popularity ranking} that
dynamically assigns the most popular products to the most visible
positions.

This paper seeks to expand our understanding of social influence in
T-O markets and explores the role of the social signal in conjunction
with product visibilities. Our starting point is the generalised
multinomial logic model of Krumme et al. \citep{krumme2012quantifying},
which we extend to vary the strength of the social signal. More
precisely, this paper considers a T-O market where the probability of
purchasing product $i$ at time $t$ is given by
\begin{equation}
\label{eq-pp}
p_i(\phi^t) = \frac{v_{\sigma(i)} q_i (\phi_i^t)^r}{\sum_{j=1}^n
  v_{\sigma(j)} q_j (\phi_j^t)^r}
\end{equation}
where $\sigma$ is a bijection from $n$ products to $n$ positions,
$v_{k}\in \R$ is the visibility of position $k$, $q_i\in \R$ is the
inherent quality of product $i$, $\phi_i^t$ is the market share of
product $i$ at time $t$, and $r > 0$ is the strength of the social
signal. As should be clear from the discussion above, prior work on
T-O markets with product visibilities (e.g.,
\cite{lerman2014leveraging,Rank,ICWSM16SI,IJCAI16SI}) focused on the
case of a linear social signal ($r = 1$). The primary objective of
this paper is to understand what happens to the T-O market when $r <
1$.

The paper contains both theoretical and simulation results and its
contributions can be summarised as follows:

\begin{enumerate}
\item When $r < 1$ and a static ranking is used, {\em the market converges
  to a unique equilibrium, which we characterise analytically}. In the
  equilibrium, the market shares depend only on the product qualities
  $q_i$ and no monopoly occurs. Moreover, a product of higher quality
  receives a larger market share than a product of lower quality,
  introducing a notion of fairness in the market and reducing the
  inequalities introduced by a linear social signal.

\item When $r > 1$ and a static ranking is used, the equilibria can be
  characterised similarly. However, contrary to the case $r < 1$, the
  equilibria that are not monopolies can be shown to be unstable under
  certain conditions. As a result, the market will typically go to a
  monopoly for some product: Which product wins the entire market
  share depends on the initial condition and the early dynamics.

\item Agent-based simulations show that the market converges quickly
  towards an equilibrium when using sublinear social signals and the
  quality ranking. They also show that the quality ranking outperforms
  the popularity ranking in maximising the efficiency of the
  market. The popularity ranking is also shown to have some
  significant drawbacks in some settings.
\end{enumerate}

\noindent
These theoretical results indicate that, when the social influence
signal is a sublinear function of the market share and a static
ranking of the products (e.g., the quality ranking) is used, the
market is entirely predictable, depends only on the product quality,
and does not lead to a monopoly. This contrasts with the case of $r=1$
where the market is entirely predictable but goes to a monopoly for
the product of highest quality (assuming the quality ranking)
\citep{ICWSM16SI} and the case of $r > 1$ where the market becomes
unpredictable (even with a static ranking). As a result, sublinear
social signals provide a way to balance market efficiency and the
inequalities introduced by social influence. In particular, with
sublinear social signals and a static ranking, markets do not exhibit
a  Matthew effect where the winner takes all, and remain predictable.

The remaining of this paper is organised as follows. Section
\ref{section:related} describes the related work. Section
\ref{section:market} introduces T-O markets and the generalised
multinomial logit model for consumer preferences considered
here. Section \ref{section:RMA} reviews some necessary mathematical
preliminaries, including the fact that T-O markets can be modeled as
Robbins-Monro algorithms.  Section \ref{section:equilibrium} derives
the equilibria for the market as a function of the social signal and
also presents the convergence results.  Section
\ref{section:simulation} reports the results from the agent-based
simulation. Section \ref{section:additional} discusses some additional
results on sublinear signals. Section \ref{section:conclusion}
discusses the results and concludes the paper.

\section{Related Work}
\label{section:related}

The research presented in this paper was motivated by the seminal work
of Salganik et al. \citep{salganik2006experimental}. They study an
experimental market called the \musiclab{}, where participants were
presented a list of unknown songs from unknown bands, each song being
described by its name and band. The participants were partitioned into
two groups exposed to two different experimental conditions: the {\em
  independent} condition and the {\em social influence} condition.  In
the independent group, participants were shown the songs in a random
order and they were allowed to listen to a song and then to download
it if they so wish.  In the second group (social influence condition),
participants were shown the songs in popularity order, i.e., by
assigning the most popular songs to the most visible
positions. Moreover, these participants were also shown a social
signal given by the number of times each song was downloaded. In order
to investigate the impact of social influence, participants in the
second group were distributed in eight ``worlds'' evolving completely
independently. In particular, participants in one world had no
visibility about the downloads and the rankings in the other
worlds. The \musiclab{} exemplifies a T-O market where each song
represents a product, and listening and downloading a song represent
trying and purchasing a product respectively. The results in
\cite{salganik2006experimental} show that different worlds evolve
differently from one another, and significantly so, providing evidence
that social influence may introduce unpredictability, inequalities,
and inefficiency in the market.

The results in \cite{salganik2006experimental} were reproduced by
numerous researchers (e.g.,
~\cite{salganik2009web,salganik2008leading,Muchnik2013,vandeRijt2014})
and, in particular, by Krumme et al. \cite{krumme2012quantifying} who
model the \musiclab{} experiment with a generalised multinomial logit
where product utilities depend on the song appeal, quality,
visibility, and a social influence signal representing past
purchases. The T-O market studied in this paper generalises the model proposed
by Krumme et al.,  exploring various strengths for the
social signal as indicated in Equation \ref{eq-pp}. The case of a
linear signal ($r=1$) has been given significant attention. Abeliuk et
al \cite{Rank} proposed the {\em performance ranking} which orders the
products optimally at each time $t$ given the appeals, qualities,
visibilities, and market shares. They show that, when the performance
ranking is used, the market always benefits from social influence in
expectation. Van Hentenryck et al. \cite{ICWSM16SI} study the quality
ranking which ranks the products by quality: They show that the
quality ranking and, more generally, any static ranking, always
benefits from social influence in expectation. They also prove that
the market converges almost surely to a monopoly for the
highest-quality product, indicating that the quality ranking is both
optimal and predictable asymptotically. These results extend
well-known theorems on P{\? o}lya urns and their generalisations
(e.g., \cite{PolyaUrnModels,Chung2003,Renlund2010}). Abeliuk et al.
\cite{IJCAI16SI} also show that the performance ranking converges to a
monopoly for the highest-quality product when a linear social signal
is used.

Ceyhan et al. \citep{Ceyhan11} study a general choice probability
$C_i^J(\phi^t)$, where $J$ represents the strength of the social
signal, and prove some general convergence results under some
assumptions. In particular, they use the ODE method \cite{Ljung77} and
a  Lyapunov function (e.g., \cite{Kushner03}) to prove that
the market converges to an equilibrium when the Jacobian of $C_i^J$ is
symmetric (which is not the case when product visibilities are
present). They also study in detail the case where the market follows
a logit model of the form
\[
C_i^J(\phi) = \frac{e^{J \phi_i + q_i}}{\sum_j e^{J \phi_j + q_j}}
\]
where $J$ is a constant capturing the strength of the social influence
signal. They show that there exists a parameter $J^*$ such that the
market converges toward a unique equilibrium when $J < J^*$ and to a
monopoly when $J \geq J^*$.  No analytical characterisation of the
equilibrium when $J < J^*$ is presented.

It is interesting to contrast these and our results. Observe first
that the proof technique used in \cite{Ceyhan11} relies on the fact
that the Jacobian of $C_i^J$ is symmetric, which is not the case for
T-O markets with product visibilities.  Our paper studies such T-O
markets and show that, when $0 < r < 1$ and a static ranking is used,
{\em the market converges to an inner equilibrium, which we
  characterise analytically.} When $r=1$, the T-O market converges to
a monopoly for the product with the highest value $v_{\sigma(i)} q_i$
\citep{Rank}. When $r > 1$, we show that the equilibria of the T-O
market are given by monopolies for each product and other type of
equilibria (e.g, a market share consisting on a distribution $60\%,
40\%, 0\%$ for a market with 3 products). {\em We prove that, when $r
  > 1$, the equilibria that are not monopolies are unstable (under
  certain conditions)}. As a result, the market will likely converge
to a monopoly for some product.

It is also useful to mention that different, theoretical and
experimental, approaches to the use of social influence are present in
the literature. For instance, Yuan and Hwarng \citep{yuan2016} describe
a demand-based pricing model under social influence and capture its
\rev{behaviour} with a dynamical system that evolve to some stable or chaotic
equilibria depending on the strength of the social signal. Stummer et
al \citep{stumm2015} introduces an agent-based model for repeat
purchase decisions addressing different types of innovation diffusion
and their perceived attributes; They also applied this methodology to
an application concerned with second-generation biofuel.

\section{The Trial-Offer Model}
\label{section:market}

The paper builds on the work by Krumme et al.
\cite{krumme2012quantifying} who propose a framework in which consumer
choices are captured by a multinomial logit model whose product
utilities depend on the product appeal, position bias, and a social
influence signal representing past purchases. A marketplace consists
of a set $N$ of $n$ items. Each item $i \in N$ is characterised by two
values:
\begin{enumerate}
\item its {\em appeal} $a_i > 0$ which represents the inherent preference
  of trying item $i$;

\item its {\em quality} $q_i > 0$ which represents the conditional
  probability of purchasing item $i$ given that it was tried.
\end{enumerate}

\noindent
This paper assumes that the appeals and the qualities are
known. Abeliuk et al. \cite{Rank} have shown that these values can be
recovered accurately and quickly, either before or during the market
execution using the approximation suggested by Krumme et al.:
\[
a_i\sim \frac{s_i}{\sum_j s_j}
\]
and
\[
q_i\sim \frac{d_i}{s_i}
\]
where $s_i$ and $d_i$ are the samplings and purchases of product $i$ at
some point in time.

The objective of the firm running this market is to maximise the total
expected number of purchases. To achieve this, the key managerial
decision of the firm is what is known as the ranking policy
\citep{Rank}, which consists in deciding how to display the products
in the market (e.g., where to display a product on a web page). Here
we assume that, at the beginning of the market, the firm decides upon
a ranking for the items, i.e., an assignment of items to positions in
the marketplace. Each position $j$ has a visibility $v_j$ which
represents the inherent probability of trying an item in position
$j$. A ranking $\sigma$ is a permutation of the items and
$\sigma(i)=j$ means that item $i$ is placed in position $j$ ($j \in
N$).  When a customer enters the market, she observes all the items
and their social signals based on the values of the previous purchases
$d^t = (d^t_1,\ldots,d^t_n)$.

The vector $\phi^t$ of market shares at time $t$ is computed in terms
of the vector $d^t$, i.e.,
\[
\phi_i^t = \frac{d_i^t}{\sum_{j=1}^n d_j^t}.
\]
and
\[
\phi^t \in \Delta^{n-1} = \{ x = (x_1,\ldots,x_n) \in \R^n \ \mid  \ 0 \leq x_i \leq 1 \mbox{ and } \sum_{i=1}^n x_i = 1 \}.
\]
The consumer then selects an item to try. The probability that the customer
tries item $i$ is given by $P_i(\sigma,\phi^t)$ where
\begin{equation}
\label{probatry}
P_i(\sigma,\phi) =  \frac{v_{\sigma(i)} f(\phi_i)}{\sum_{j=1}^n v_{\sigma(j)} f(\phi_j)}
\end{equation}
and $f$ is a continuous, positive, and nondecreasing function. This
probability generalises the multinomial logit model of Krumme et al.
\cite{krumme2012quantifying} who define two sets of probabilities,
$p^{SI}_{i,t}$ and $p^{I}_i$, that capture the probability of trying
product $i$ with and without social influence. These probabilities are
defined as:
\begin{equation}
\label{probkrum}
p^{SI}_{i,t}= \frac{v_{\sigma(i)}  (\alpha a_i+d_i^t)}{\sum_{j=1}^n
  v_{\sigma(j)} (\alpha a_j+d_j^t)},\qquad p_i^{I}=\frac{v_{\sigma(i)}a_i }{\sum_{j=1}^n
  v_{\sigma(j)}a_j},
  \end{equation}

\noindent
where $\alpha$ is a parameter to calibrate the strength of the social
signal (e.g., $\alpha=200$ for the \musiclab{} experiments).  Equation
\eqref{probatry} allows us to recover the formulae \eqref{probkrum}
via some linear transformation of the identity function:
$f(\phi_i)=\beta \phi_i+\alpha a_i, $ with $\beta=\sum_jd_j$ or
$\beta=0$ for each case.

After having tried product $i$, a customer decides whether to buy
the sampled item and the probability that she purchases item $i$ is
given by $q_i$.  If item $i$ is purchased at time $t$, then the
purchase vector becomes
\begin{equation*}
 d_j^{t+1} = \left\{
    \begin{array}{ll} d_j^t + 1 & \mbox{ if } j = i; \\
                      d_j^t     & \mbox{ otherwise.}
    \end{array} \right.
\end{equation*}
\noindent
For simplicity, the vector $d^0$ is initialised with the product
appeals\footnote{The initialisation can be justified by viewing the
  discrete dynamic process as an urn and balls process, where the
  appeals are the initial sets of balls.}, i.e., $d_i^0 = a_i$.  This
condition can be relaxed and the results presented in this paper
continue to hold but the notations are heavier (See Appendix
\ref{appendix:purchases} for the details). To analyse this process, we
divide time into discrete periods such that each new period begins
when a product has been purchased. Hence, the length of each time
period is not constant.\footnote{In Section 6, we  modify the notion of
  time period to analyse the efficiency of trial-offer markets.}


In this paper, we are interested in characterising how the market
shares $\{\phi^t\}_{t>0}$ evolve over time for various functions $f$
given a static ranking $\sigma$. We are particularly interested in
study the asymptotic \rev{behaviour} of $\{\phi^t\}_{t>0}$ for the cases
where $f(x)=x^r$ with $r>0$. For instance, when $r=0.5$, the social
signal displays the square root of the number of past purchases. For
notational simplicity, we assume that the ranking is fixed and is the
identity function $\sigma(i) = i$ and omit it from the formulas. If
the qualities and visibilities also satisfy $q_1 \geq \ldots \geq q_n$
and $v_1 \geq \ldots \geq v_n$, we obtain the quality ranking proposed
in \cite{ICWSM16SI} but our results hold for any static ranking.

The following lemma, whose proof is in Appendix \ref{proofs}, relates
the two phases of the T-O market and characterises the probability
that the next purchase is item $i$. It generalises a result in
\citep{ICWSM16SI}.

\begin{lemma}
  \label{probab} If $f:[0,1]\to \R$ is a positive function, then
  the probability $p_i(\phi)$ that the next purchase is the product
  $i$ given the market share vector $\phi$ is given by
\begin{equation}
p_i(\phi) \ = \frac{v_i q_i f(\phi_i)}{\sum_{j=1}^n v_j q_j f(\phi_j)} \label{probanext}.
\end{equation}
\end{lemma}

\rev{We finish this section with an important equivalence that arises when $f(x)=x^r$, $r>0$. Under such condition, if one writes down Equation \eqref{probanext} in terms of the number of purchases $d_i$ that product $i$ has so far obtained instead of using the market shares, we get exactly the same expression. Indeed,}

\begin{equation}
\rev{p_i(\phi)=\frac{v_iq_i(\phi_i)^r}{\sum_jv_jq_j(\phi_j)^r}=\frac{v_iq_i(\frac{d_i}{\sum_kd_k})^r}{\sum_jv_jq_j(\frac{d_j}{\sum_kd_k})^r}=\frac{v_iq_i(d_i)^r}{\sum_jv_jq_j(d_j)^r}.}
\end{equation}

\rev{Thus, when the social signal function $f$ is given by $f(x)=x^r$, $r>0$, it is possible to interpret our model either using the concept of market share (as we describe in Section 3) or simply using the current number of purchases.}

\section{Trial-Offer Markets as Robbins-Monro Algorithms}
\label{section:RMA}

This section establish some basic definitions and concepts that are
useful in the rest of the paper. In particular, it shows that T-O
markets can be modeled as Robbins-Monro algorithms and states some
useful results. The results in this section are well-known in
stochastic approximation. The section starts with a brief introduction
of Ordinary Differential Equations (ODE) and some stability criteria
(e.g., see \cite{Hirsch2012,jordan1999nonlinear}).

\paragraph{Differential Equations:}

\noindent
Consider the following differential equation
\begin{equation}
\label{cauchy}
\frac{dy}{dt}= F(y)
\end{equation}
where $F$ is some vector field. The concept of equilibrium is central
in the study of asymptotic \rev{behaviour} for this type of equation.

\begin{definition}
\label{stab}
A vector $y^*\in\R^n$ is an equilibrium for differential equation \eqref{cauchy} if $F(y^*)=0$.
\end{definition}

\noindent
We are interested in equilibria that satisfy (at least) certain
stability criteria.


\begin{definition}
 An equilibrium $y^*$ is said to be {\it stable} for Equation
  \eqref{cauchy} if, given $\epsilon > 0$, there exists $\delta > 0$
  such that $\| y(t) - y^* \| < \epsilon$ for all $t > 0$ and for
  all $y$ such that $\|y - y^*\| < \delta$.
  We say that $y^*$ is {\it \rev{asymptotically} stable} if \rev{it} also satisfies \[
\lim_{t\to\infty}y(t)=y^*.
\]
\end{definition}

\begin{remark}
When an equilibrium $y^*$ is not stable, we say that $y^*$ is unstable.
\end{remark}

\noindent
The asymptotic stability of an equilibrium $y^*$ can be characterised
in terms of the Jacobian matrix $JF(y^*)=(\frac{\partial
  F_i(y^*)}{\partial y_j})_{i,j}$ $(i,j \in N)$ as stated by the
following Theorem (see, for instance, \cite{jordan1999nonlinear}
p. 440).

\begin{theorem}
\label{ODEunst}
Let $y^*$ be an equilibrium for the differential equation
\eqref{cauchy}. If the eigenvalues of the Jacobian matrix $JF(y^*)$
all have negative real part, then $y^*$ is asymptotically
stable. If, on the other hand, $JF(y^*)$ has at least one eigenvalue
with a positive real part, then $y^*$ is unstable.
\end{theorem}

\noindent
A well-known result from linear algebra (see, for example,
\cite{robinson2006algebra} p. 296) establishes the connection between
the trace of a square matrix and its eigenvalues. In the following, we
use $tr(A)=\sum_{i}^n a_{ii}$ to denote the trace of matrix $A$ where
$a_{ii}$ are the diagonal entries of matrix $A$.

\begin{theorem}
\label{trace}
Let $A$ be a $n\times n$ matrix, and $\lambda_1,..., \lambda_n$ its
eigenvalues. Then $tr(A)=\sum_{i=1}^n\lambda_i$.
\end{theorem}

\noindent
We now show that the discrete stochastic process $\{\phi^t\}_{t \geq
  0}$ can be modeled as a Robbins-Monro Algorithm (RMA)
\citep{Kushner03,Duflo97}.

\begin{definition}[Robbins-Monro Algorithm] A  {\it Robbins-Monro
    Algorithm} (RMA) is a discrete time stochastic processes
  $\{x^t\}_{t \geq 0}$  whose general structure is specified by
\begin{equation}
\label{RMA}
x^{k+1}-x^k=\gamma^{k+1}[F(x^k)+U^{k+1}],
\end{equation}
\noindent	
where
\begin{itemize}
\item $x^k$ takes its values in some Euclidean space (e.g., $\R^n$);
\item $\gamma^k$ is deterministic and satisfies $\gamma^k>0$,
  $\sum_{t \geq 1}\gamma^t=\infty$, and $\lim_{t \to\infty}\gamma^t=0$ with probability 1;
\item $F:\R^n\to\R^n$ is a deterministic continuous vector field;
\item $\EE[U^{k+1}| \Fk]=0$, where $\Fk$ is the natural filtration of
  the entire process.\footnote{$\Fk$, the natural filtration, is the
    $\sigma$-field generated by the history $\{x^l: l\leq k\}$ }.
\end{itemize}
\end{definition}

\noindent
A RMA $\{x^t\}_{t \geq 0}$ where $x^t$ has $n$ coordinates is said to
be $n$-dimensional. Recall that the probability that the next
purchase is item $i$ at time $k$ is given by $p_i(\phi^k)$ and denote
by $e^k$ the random unit vector whose $j^{th}$ entry is 1 if item $j$
is the next purchase and 0 otherwise. The market share at time $k+1$
is given by
\[
\phi^{k+1} = \frac{D^k \phi^k}{D^k + 1} + \frac{e^k}{D^k + 1}
\]
where $D^k = \sum_{t=0}^k \sum_{i=1}^n d_i^t=k+k_0, \text{where } k_0=\sum_{i=1}^na_i$. It follows that
\begin{align*}
\phi^{k+1} & =  \frac{(D^k+1) \phi^k}{D^k + 1} - \frac{\phi^k}{D^k + 1}  + \frac{e^k}{D^k + 1} \\
          & =  \phi^k + \frac{1}{D^k + 1} (e^k - \phi^k) \\
          & =  \phi^k + \frac{1}{D^k + 1} (\EE[e^k | \Fk] - \phi^k + e^k - \EE[e^k | \Fk]) \\
          & =  \phi^k + \frac{1}{D^k + 1} (p(\phi^k) - \phi^k + e^k - \EE[e^k | \Fk]).
\end{align*}
This last equality can be reformulated as
\begin{equation}
\label{RMA-MS}
\phi^{k+1} = \phi^k + \gamma^{k+1} (F(\phi^k) + U^{k+1})
\end{equation}
where $\gamma^{k+1} = \frac{1}{D^k + 1}$, $F(\phi) = p(\phi) - \phi$,
and $U^{k+1} = e^k - \EE[e^k | \Fk]$. Note that the function $F$
captures the difference between the probabilities of purchasing the
items (given the market shares) and the market shares at each time
step. Recall that $\phi^k\in\Delta^{n-1}$ for all $k\geq 0$, which is
a compact, convex subset of $\R^n$ (and hence connected). We can now
prove that the discrete dynamic process $\{\phi^t\}_{t \geq 0}$ can be
modeled as a Robbins-Monro algorithm.

\begin{theorem} The discrete stochastic dynamic process
  $\{\phi^k\}_{k\geq0}$ can be modeled as the Robbins-Monro algorithm.
\end{theorem}
\begin{proof}
  The above derivation showed that $\{\phi^k\}_{k \geq 0}$ can be
  expressed through Equation \eqref{RMA-MS}, i.e.,
\[
\phi^{k+1} = \phi^k + \gamma^{k+1} (F(\phi^k) + U^{k+1})
\]
where $\gamma^{k+1} = \frac{1}{D^k + 1}=\frac{1}{k+k_0+1}$, $F(\phi) = p(\phi) - \phi$,
and $U^{k+1} = e^k - \EE[e^k | \Fk]$. It is easy to see that
$\gamma^k>0$, $\sum_{k \geq 1}\gamma^k=\infty$, $\lim_{k
  \to\infty}\gamma^k=0$, and that  $\EE[U^{k+1}|\Fk]$ is
equal to zero.
\end{proof}

\noindent
Robbins-Monro algorithms are particularly interesting because, under
certain conditions on $x^k$, $\gamma^k$, and $U^{k+1}$, their asymptotic
\rev{behaviour}, i.e., the values of $x^k$ when $k\to\infty$, is closely
related to the asymptotic \rev{behaviour} of the following continuous dynamic
process:
\begin{equation}
\label{RMC}
\frac{dx^t}{dt}=F(x^t), \quad x^t\in\Delta^{n-1}.
\end{equation}

\noindent
This idea, called {\it the ODE Method}, was introduced by
\cite{Ljung77} and has been extensively studied (e.g.,
\citep{Borkar00,Duflo97,Kushner03}). Consider again the RMA $\{x^k\}_{k \geq
  0}$ defined in (\ref{RMA}) and the following hypotheses:
\begin{itemize}
\item[H1:] $\dis \sup_k\EE[\|U^{k+1}\|^2]<\infty$;
\item[H2:] $\dis \sum_k(\gamma^k)^{2}<\infty$;
\item[H3:] $\dis \sup_k\|x^k\|<\infty$.
\end{itemize}

\noindent
We will now present a theorem establishing the connection between the
discrete stochastic process \eqref{RMA} and the continuous process
defined by \eqref{RMC}. This connection requires the concept of
Internally Chain Transitivity (ICT) sets. These ICT sets include
equilibria, periodic orbits of (\ref{RMC}), and possibly more
complicated sets.

To define ICT sets formally for the purpose of this paper, we use
Proposition 5.3 in \cite{Benaim99} that proves that the concepts of
internally chain recurrent and internally chain transitive set are
equivalent when the set over which $F$ is defined is connected, which
is obviously the case here.

\begin{definition}[$(\epsilon,T)$-Chains \cite{conley1978isolated}]
  Consider $\epsilon>0$, $T>0$, a set $A\subset\R^n$, and two points
  $x,y\in A$. There is an $(\epsilon,T)$-chain of length $k$ in $A$ between $x$ and
  $y$ if there exist $k$ solutions $\{y_1,...,y_k\}$ of \eqref{RMC}
  and their associated times $\{t_1,...,t_k\}$ with $t_i > T$ such
  that
\begin{enumerate}
\item[1.] $y_i^t\in A$  for all $t\in [0,t_i]$ and for all $i\in\{1,...,k\};$
\item[2.] $\|y_i^{t_i}-y_{i+1}^0\|< \epsilon$  for all $i\in\{1,...,k-1\};$
\item[3.] $\|y_1^0-x\|<\epsilon$ and $\|y_k^{t_k}-y\|<\epsilon.$
\end{enumerate}
\end{definition}

\noindent
We are now in a position to define ICT sets, which is derived from the
definition of Internally Chain Recurrent sets introduced by Conley
(1978) \cite{conley1978isolated}.

\begin{definition}[ICT Sets]
  A closed set $A$ is said {\it Internally Chain Transitive (ICT)} for
  the dynamics \eqref{RMC} if it is compact, connected, and for all
$\epsilon>0, T>0$ and $x,y\in A$, there exists an $(\epsilon,T)$-chain
in $A$ between $x$ and y.
\end{definition}


%
%

\noindent
The following theorem, due to \cite{Benaim99} and whose proof is in
Appendix \ref{proofs}, links the \rev{behaviour} of the limit set
$L\{x^k\}_{k \geq 0}$ of any sample path $\{x^k\}_{k \geq 0}$ for
Equation \eqref{RMA} and the limit sets of the solution to Equation
\eqref{RMC}.

\begin{theorem}[\cite{Benaim99}]
\label{thm:ICT}
Let $\{x^k\}_{k \geq 0}$ be a Robbins-Monro algorithm \eqref{RMA}
satisfying hypotheses $H1- H3$ where $F$ is a bounded locally Lipschitz vector field (e.g., a bounded
$\mathcal{C}^{1}$ function). Then, with probability 1, the limit set
$L\{x^k\}_{k\geq 0}$ is internally chain transitive for Equation
\eqref{RMC}.
\end{theorem}


\noindent
Note that Theorem \ref{thm:ICT} is valid for very general functions
$F(x)=p(x)-x$, as the only requirement is to be locally Lipschitz.
Since hypotheses H1, H2, and H3 are all satisfied by the discrete
stochastic dynamic process $\{\phi^t\}_{t\geq0}$, it remains to study
the structure of the ICT set of Equation \eqref{RMC}. This paper
focuses on the cases where the social signal $f(x)$ is given by
$f(x)=x^r$, with $r>0$. We will show that the ICT set of Equation
\eqref{RMC} only consists of equilibria.

\section{Equilibria of Trial-Offer Markets}
\label{section:equilibrium}

This section characterises the equilibria and the asymptotic \rev{behaviour} of the continuous dynamics
\begin{equation}
\label{RMC-MS}
\frac{d\phi^t}{dt}= p(\phi^t) - \phi^t, \quad (\phi^t \in \Delta^{n-1}),
\end{equation}
which is associated with the RMA \eqref{RMA-MS}. For simplicity, we remove the
visibilities by stating $\overline{q}_j = v_jq_j$. We are interested in the case where $f(x)=x^r$ with $(r>0, r\neq 1)$, since the case
$r=1$ has been settled in earlier work. \rev{Let $Q$ be  the set of positive market shares, this is, $Q=\{i\in N: \phi_i\neq 0\}$, clearly $Q\neq \emptyset$ since $\sum_{i=1}^n\phi_i=1$}.
	
\begin{theorem}
\label{eqs}
Let $f(x)=x^r, r>0$, and $r\neq 1$. Any equilibria $\phi$ for Equation \eqref{RMC-MS} has coordinates
\[
\phi_i=\dfrac{\overline{q}_i^{\frac{1}{1-r}}}{\sum_{j\in Q}
  \overline{q}_j^{\frac{1}{1-r}}} \quad \mbox{ if } i\in Q 
\]
and zero otherwise (i.e., if $i \in N\setminus Q$).
\end{theorem}

\begin{proof}
  An equilibrium to \eqref{RMC-MS} must satisfy $p_i(\phi)=\phi_i$,
  i.e.,
 \[
\frac{\overline{q}_i(\phi_i)^r}{\sum_{j=0}^n\overline{q}_j(\phi_j)^r}=\phi_i.
\]
For $i\in Q$, we have
\[
\overline{q}_i(\phi_i)^{r-1}=\sum_{j\in Q}\overline{q}_j(\phi_j)^r
\]
and, for all $i , k \in Q$, we also have
\[
 \overline{q}_i(\phi_i)^{r-1}=\sum_{j\in Q}\overline{q}_j(\phi_j)^r= \overline{q}_k(\phi_k)^{r-1}
\]
which is equivalent to
\begin{equation}
\label{eq:equilibrium}
\overline{q}_i(\phi_i)^{r-1}=\overline{q}_k(\phi_k)^{r-1} \Leftrightarrow \phi_i=\left(\frac{\overline{q}_k}{\overline{q}_i}\right)^{\frac{1}{r-1}}\phi_k.
\end{equation}
\noindent
By summing for all $i\in Q$, we obtain
\[
1=\sum_{i\in Q}\phi_i=\frac{\phi_k}{\overline{q}_k^{1/(1-r)}}\sum_{i\in Q}\overline{q}_i^{1/(1-r)}
\]
and hence
\[
\phi_k=\frac{\overline{q}_k^{1/(1-r)}}{\sum_{i\in Q}\overline{q}_i^{1/(1-r)}}.
\]
It remains to prove $\phi$ is indeed an equilibrium, i.e.,
$p(\phi)=\phi$. This is equivalent to prove that
$p_i(\phi)=\phi_i$ for all $i \in \{1,\ldots,n\}$. The result is
obvious if $i \in N\setminus Q$ ($\phi_i=0 \Rightarrow$ $p_i(\phi)=0$). If $i \in Q$, then
\begin{align*}
p_i(\phi)&=\frac{\overline{q}_i(\phi_i)^r}{\sum_{j\in Q}\overline{q}_j(\phi_j)^r}\\\
&=\frac{\overline{q}_i(\oq_i^{1/(1-r)})^r}{\sum_{j\in Q}\overline{q}_j(\oq_j^{1/(1-r)})^r}*\frac{\left(\sum_{j\in Q}\oq_j^{1/(1-r)}\right)^r}{\left(\sum_{j\in Q}\oq_j^{1/(1-r)}\right)^r}\\
&=\frac{\oq_i^{[1+r/(1-r)]}}{\sum_{j\in Q}\oq_j^{[1+r/(1-r)]}}\\
&=\frac{\oq_i^{1/(1-r)}}{\sum_{j\in Q}\oq_j^{1/(1-r)}}=\phi_i.
\end{align*}
\end{proof}
\noindent	
Note that, when $|Q|=n$, the equilibrium lives in the interior of the
simplex $int(\Delta^{n-1})$ (all its coordinates are strictly
positive). We use $\phi^*$ to denote this equilibrium. When $|Q|=1$,
then the equilibrium is one of the vertices of the simplex. Finally,
the cases $1<|Q|<n$ cover the other possible equilibria (for example
$\phi=(3/4,1/4,0,\dots,0)$).

Observe also that the equilibrium $\phi^*\in int(\Delta^{n-1})$ for the
case $0<r<1$ has some very interesting properties: It is unique, which
means that the market is completely predictable.  Moreover, if
$\overline{q}_i\geq \overline{q}_j$, then $\phi_i^*\geq \phi_j^*$,
which endows the market with a basic notion of {\it
  fairness}. Finally, the market is not a monopoly: All the market
shares are strictly positive for the equilibrium $\phi^*$.
	

Our next result characterises the ICT of the continuous dynamics. We
start with a useful lemma which indicates that submarkets can also be
modeled as RMAs.
	
\begin{lemma}
\label{induction}
Consider a T-O market defined by $n$ items and the submarket
obtained by considering only $n-1$ items.  Then  this submarket can also be
modeled by an RMA.
\end{lemma}

\begin{proof}
  Let $\Phi^t=[\phi_1^t,\phi_2^t,\cdots,\phi_n^t]$ be the market share
  for the $n$-item T-O market at stage $t$.  Consider a new process
  $\{\Psi^t\}_{t \geq 0}$ consisting of $n-1$ products only. We show
  that this process can also be modeled as a RMA. The key is to prove
  that the probability of purchasing product $j$ in stage $t$ follows
  Equation \eqref{probanext}. Consider any item $i\in\{1,..,n\}$ such
  that $\phi_i^t\neq 1$. Without loss of generality, assume that
  $i=n$, define
\[
\psi_i^t=\frac{\phi_i^t}{1-\phi_n^t},\quad (i<n),
\]
and consider the following events:
\begin{itemize}
\item $A=\{\mbox{product } n \mbox{ is not purchased at stage } t\}$
\item $B=\{\mbox{product } j\neq n \mbox{ is purchased at stage } t\}$.
\end{itemize}

\noindent
Since $B \subseteq A$, $Pr[B\cap A]=Pr[B]=\dfrac{\overline{q}_j(\phi_j^t)^r}{\sum_{i=1}^n\overline{q}_i(\phi_i^t)^r}$. On the other hand
\[
Pr[A]=1-\dfrac{\overline{q}_n(\phi_n^t)^r}{\sum_{i=1}^n\overline{q}_i(\phi_i^t)^r}=\dfrac{\sum_{j=1}^{n-1}\overline{q}_j(\phi_j^t)^r}{\sum_{i=1}^n\overline{q}_i(\phi_i^t)^r},
\]
and therefore
\[
Pr[B| A]=\frac{Pr[B\cap A]}{Pr[A]}=\dfrac{\overline{q}_j(\phi_j^t)^r}{\sum_{i=1}^{n-1}\overline{q}_i(\phi_i^t)^r}\cdot \dfrac{(1-\phi_n^t)^r}{(1-\phi_n^t)^r}=\dfrac{\overline{q}_j(\psi_j^t)^r}{\sum_{i=1}^{n-1}\overline{q}_i(\psi_i^t)^r}.
\]
Since $\psi_i^t\geq 0$ and
$\sum_{i=1}^{n-1}\psi_i^t=\sum_{i=1}^{n-1}\frac{\phi_i^t}{1-\phi_n^t}=\frac{1}{1-\phi_n^t}\sum_{i=1}^{n-1}\phi_i^t=1$,
the $\psi_i^t$ are well-defined market shares. Since the evolution of $\psi^t$ depends on the probability $Pr[B|A]$, one can obtain a similar formula to \eqref{RMA-MS}. Indeed,  observe that on the event $A$, we have that for every $i=1,\hdots,n-1$, $\psi_i^{k+1}  =  \frac{(D^k+1) \psi_i^k}{D^k + 1} - \frac{\psi_i^k}{D^k + 1}  + \frac{\hat{e}_i^k}{D^k + 1} $, with $\EE[\hat{e}^k|\Fk]=Pr[B|A]$. Hence, 
$\{\psi^t\}_{t\geq0}$ can be modeled as \rev{an} $n-1$ dimensional RMA.
\end{proof}

\noindent
We are now in position to prove the main result of this paper. The
theorem considers the case where $\phi^0\in int(\Delta^n)$, which is
the case when the product appeals are strictly positive. It proves
that, under this condition, the ICT set of the RMA $\{\phi^t\}_{t>0}$
consists of a single equilibrium $\phi^*$.

\begin{theorem}
\label{thm:ict}
Under the social signal $f(x)=x^r, 0<r<1$ with $\phi^0\in
int(\Delta^n)$, the RMA $\{\phi^t\}_{t>0}$ converges to $\phi^*$
almost surely.
\end{theorem}
\begin{proof}
The proof studies the asymptotic \rev{behaviour} of the solutions of the following ODE:
\begin{equation}
\label{ODEE}
\frac{d \phi^t}{dt}=p(\phi^t)-\phi^t.
\end{equation}
Equation (\ref{ODEE}) is equivalent to
\[
\frac{d \phi_i^t}{dt}=\frac{\overline{q}_i(\phi_i^t)^r}{\sum_j\overline{q}_j(\phi_j^t)^r}-\phi_i^t,\quad i\in\{1,\cdots,n\};  0<t<\infty.
\]
Hence, we have
\begin{align*}
&\frac{\overline{q}_i(\phi_i^t)^r}{\sum_j\overline{q}_j(\phi_j^t)^r}=\frac{d\phi_i^t}{d t}+\phi_i^t, \\
&\frac{1}{\sum_j\overline{q}_j(\phi_j^t)^r}=\frac{1}{\overline{q}_i(\phi_i^t)^r}[\frac{d\phi_i^t}{d t}+\phi_i^t] \\
&\frac{1}{\overline{q}_i(\phi_i^t)^r}[\frac{d\phi_i^t}{d t}+\phi_i^t]=\frac{1}{\overline{q}_j(\phi_j^t)^r}[\frac{d\phi_j^t}{d t}+\phi_j^t] \quad \forall i,j\in\{1,\cdots,n\}, \\
&\overline{q}_i^{-1}(\phi_i^t)^{-r}[\frac{d\phi_i^t}{d t}+\phi_i^t]=\overline{q}_j^{-1}(\phi_j^t)^{-r}[\frac{d\phi_j^t}{d t}+\phi_j^t], \\
&\overline{q}_i^{-1}[(\phi_i^t)^{-r}\frac{d\phi_i^t}{d t}+(\phi_i^t)^{1-r}]=\overline{q}_j^{-1}[(\phi_j^t)^{-r}\frac{d\phi_j^t}{d t}+(\phi_j^t)^{1-r}], \\
&e^{(1-r)t}(1-r)\overline{q}_i^{-1}[(\phi_i^t)^{-r}\frac{d\phi_i^t}{d t}+(\phi_i^t)^{1-r}]=e^{(1-r)t}(1-r)\overline{q}_j^{-1}[(\phi_j^t)^{-r}\frac{d \phi_j^t}{d t}+ (\phi_j^t)^{1-r}], \\
&\frac{d}{dt}\left[ e^{(1-r)t}\overline{q}_i^{-1}(\phi_i^t)^{1-r}\right]=\frac{d}{dt}\left[ e^{(1-r)t}\overline{q}_j^{-1}(\phi_j^t)^{1-r}\right]
\end{align*}
\noindent where the fourth equivalence is obtained by multiplying both sides with $\mu(t)=(1-r)e^{(1-r)t}$.  Notice also that as $\phi_i^0>0$, then for any finite time $t>0$, $\phi_i^t>0$.  Taking the integral $\int_0^tdt$ of the last expression gives
\begin{equation}
e^{(1-r)t}\overline{q}_i^{-1}(\phi_i^t)^{1-r}-\overline{q}_i^{-1}(\phi_i^0)^{1-r}=e^{(1-r)t}\overline{q}_j^{-1}(\phi_j^t)^{1-r}-\overline{q}_j^{-1}(\phi_j^0)^{1-r}
\end{equation}
and hence
\begin{equation}
\label{equili}
\frac{(\phi_i^t)^{1-r}}{\overline{q}_i}-\frac{(\phi_j^t)^{1-r}}{\overline{q}_j}= e^{(r-1)t}\left[   \frac{(\phi_i^0)^{1-r}}{\overline{q}_i}-\frac{(\phi_j^0)^{1-r}}{\overline{q}_j} \right].
\end{equation}
Consider Equation \eqref{equili}:
\begin{itemize}
\item  if, for some $i\neq j,$ $\displaystyle \frac{(\phi_i^0)^{1-r}}{\overline{q}_i}=\frac{(\phi_j^0)^{1-r}}{\overline{q}_j}$, then $\dis\frac{(\phi_i^t)^{1-r}}{\overline{q}_i}=\frac{(\phi_j^t)^{1-r}}{\overline{q}_j}, $  for all  $t$;
\item if $\displaystyle
  \frac{(\phi_i^0)^{1-r}}{\overline{q}_i}\neq\frac{(\phi_j^0)^{1-r}}{\overline{q}_j}$,
  then the right-hand side of Equation \eqref{equili} goes to zero as $t\to \infty$  (because
  $r<1$) and hence the left-hand side  of \eqref{equili}  also goes to zero:
\begin{equation}
\label{limit}
\lim_{t\to\infty} \frac{(\phi_i^t)^{1-r}}{\overline{q}_i}-\frac{(\phi_j^t)^{1-r}}{\overline{q}_j}=0.
\end{equation}
\end{itemize}

\noindent
We now prove by induction that the limits for the market shares
exist. Consider first the case of 2 products. Since $\phi_2^t =
(1-\phi_1^t)$, the market is completely characterised by the value of
$\phi_1^t$ and hence we can use a one-dimensional RMA and, by Theorem
1 in \citep{Renlund2010}, the RMA converges since $F(x)=p(x)-x$ is a continuous
function and $\phi_1^t$ is bounded. Assume now that a RMA with $k-1$
products converges and consider a market with $k$ products.  By Lemma
\ref{induction}, given a $k$-dimensional RMA
$\Phi^t=[\phi_1^t,\phi_2^t,\cdots,\phi_k^t]$, we can create a $k-1$
dimensional RMA $\{\Psi^t\}_{t\geq 0}$ given by
$\psi_i^t=\frac{\phi_i^t}{1-\phi_k^t} \quad (i<k)$. By induction
hypothesis, $\psi_i=\lim_{t\to\infty}\psi_i^t$ exists for all $i<k$
and therefore Equation \eqref{equili} is equivalent to
 \begin{equation}
\label{equili2}
\frac{(\phi_k^t)^{1-r}}{\overline{q}_k(1-\phi_k^t)^{1-r}}-\frac{(\psi_i^t)^{1-r}}{\overline{q}_i}= \frac{e^{(r-1)t}}{(1-\phi_k^t)^{1-r}}\left[   \frac{(\phi_k^0)^{1-r}}{\overline{q}_k}-\frac{(\phi_i^0)^{1-r}}{\overline{q}_i} \right].
\end{equation}
Observe that, if $\lim_{t\to\infty}\phi_k^t=1$, then
$\lim_{t\to\infty}\phi_j^t=0$ for all $j\neq k$, and the market shares
converge to one of the possible equilibria (i.e., a monopoly of the
product $k$). Otherwise, the right-hand side of \eqref{equili2} goes
to 0 when $t\to\infty$ and $\dis
\lim_{t\to\infty}\frac{(\psi_i^t)^{1-r}}{\overline{q}_i} $
exists. Hence $\dis
\lim_{t\to\infty}\frac{(\phi_k^t)^{1-r}}{\overline{q}_k(1-\phi_k^t)^{1-r}}$
also exists.

Now denote by $\phi_j$ the limit of $\phi_j^t$ for all $j \in
\{1,\cdots,n\}$. Using Equation \eqref{limit}, the following equation holds for all $i,j
\in\{1,\cdots,n\}$:
\begin{equation}
\label{holds}
\frac{\phi_{i}^{1-r}}{\overline{q}_i}=\frac{\phi_{j}^{1-r}}{\overline{q}_j}.
\end{equation}
\noindent
Observe that, if there exists $l\in \{1,\cdots,n\}$ such that $\phi_l=0$,
Equation \eqref{holds} implies that $\phi_i=0$ for all $i$ which is
impossible since they add up to 1. Hence the limit process has strictly
positive components and Equation \eqref{holds} is equivalent to
\begin{align}
\label{holds1}
\phi_{i}=\frac{\phi_{j}}{\overline{q}_j^{1/(1-r)}}\overline{q}_i^{1/(1-r)}
\end{align}
which is the equation that defines $\phi^*$ in Theorem \ref{eqs} (see
Equation (\ref{eq:equilibrium})).  As a result, when $\phi^0\in
int(\Delta^{n-1})$, the only ICT set for the ODE (\ref{ODEE}) is the
equilibrium $\phi^*$ and, by Theorem \ref{thm:ICT}, the RMA converges
almost surely to $\phi^*$.
\end{proof}

	

\noindent
Consider now the case $r>1$ for which Theorem \ref{eqs} still
characterises the equilibria. In this case, the dynamic \rev{behaviour} is
completely different due to the strength of the social signal. 
It is however possible to prove
that the ICT set of the RMA $\{\phi^t\}_{t>0}$ consists only of
equilibria.

\begin{theorem}
  Consider the social signal $f(x)=x^r$ with $r>1$. The RMA
  $\{\phi^t\}_{t \geq 0}$ converges almost surely to one of the
  equilibria $\phi\in Z_F:=\{x\in \Delta^{n-1}: p(x)-x=0\}$.
\end{theorem}
	
\begin{proof}
  The analysis of the ODE is the same as in Theorem \ref{thm:ict}
  since the only restriction in the proof is $r\neq1$. However, the
  interpretation of Equation (\ref{equili}) changes when $r>1$.

  \noindent We define $H_{i,0}:=\dfrac{(\phi_{i}^0)^{1-r}}{\overline{q}_i}$ for all $1\leq i\leq n$, and
  order the products in decreasing order of $H_{i,0}$. Let
  $h:\{1,..,n\}\to \{1,..,n\}$ be the permutation that defines this
  order and denotes by $h^{-1}$ its inverse function, i.e.,
  $h^{-1}(i)=j$ means that product $j$ is in the $i$-th position in
  permutation $h$.  We have that $H_{h^{-1}(1),0}\geq \cdots\geq
  H_{h^{-1}(n),0}$. Define the following sets:
\begin{itemize}
\item $\dis Q_0=\{ i \in \{1,..,n-1\}: H_{h^{-1}(i),0}=H_{h^{-1}(i+1),0}\},$
\item $\dis Q_1=\{i\in \{1,..,n-1\}: H_{h^{-1}(i),0}>H_{h^{-1}(i+1),0}\},$
\end{itemize}

\noindent
and consider the following case analysis:

\begin{itemize}
\item[i)] If $|Q_0|=n-1$, then $H_{h^{-1}(i),0}=H_{h^{-1}(i+1),0}$ for
  all $1\leq i\leq n-1$. By Equation \eqref{equili},
  $\dis\frac{(\phi_{h^{-1}(i)}^t)^{1-r}}{\overline{q}_{h^{-1}(i)}}=\frac{(\phi_{h^{-1}(i+1)}^t)^{1-r}}{\overline{q}_{h^{-1}(i+1)}},
  $ for all $t>0$ and for all $1\leq i\leq n-1$, which leads again to the inner
  equilibrium $\phi^*$.

\item[ii)] If $0<|Q_0|<n-1$, select $i\notin Q_0$. Equation
  \eqref{equili} implies that
\[\lim_{t\to\infty} \frac{(\phi_{h^{-1}(i)}^t)^{1-r}}{\overline{q}_{h^{-1}(i)}}-\frac{(\phi_{h^{-1}(i+1)}^t)^{1-r}}{\overline{q}_{h^{-1}(i+1)}}=\infty, \]
because $r>1$ and hence $e^{(r-1)t}\to \infty$ when $t\to \infty$. It
follows that $\dis \lim_{t\to\infty} \phi_{h^{-1}(i)}^t=0$ and the RMA
necessarily converges to one of the equilibria that live in the
boundary of the simplex, but they are not monopolies (see Theorem \ref{eqs}).

\item[iii)] If $|Q_0|=0$ then $|Q_1|=n-1$, Using a similar reasoning
  as in case ii), it follows that $\dis \lim_{t\to\infty}
  \phi_{h^{-1}(i)}^t=0$ for all $1\leq i\leq n-1$ and, since
  $\phi^t\in\Delta^{n-1}$ for all $t$, $\dis \lim_{t\to\infty}
  \phi_{h^{-1}(n)}^t=1$.
\end{itemize}

\noindent
As a result, the only ICT for the differential equation \eqref{ODEE}
are equilibria and, by Theorem \ref{thm:ICT}, the RMA
$\{\phi^t\}_{t\geq0}$ converges almost surely to one of them.
\end{proof}

\noindent
It is important to observe that, in the case $r>1$, the initial
conditions, i.e., the initial appeals and how the market evolves early
on, affect the entire dynamics. This is in contrast with the case
$r<1$ for which the long-term \rev{behaviour} only depends of the product
qualities. This has fundamental consequences for the predictability
and efficiency of the market. We will show that, when $r>1$, the inner
equilibrium $\phi^*$ is always unstable. The result will follow as
corollary of the following theorem.

\begin{theorem}
\label{theorem:trace}
Consider  the equilibria given by \[
\phio_i=\dfrac{\overline{q}_i^{\frac{1}{1-r}}}{\sum_{j\in Q}
  \overline{q}_j^{\frac{1}{1-r}}} \; \mbox{ if } i\in Q \; \mbox{ and } \; \phio_i=0 \; \mbox{ if } i\in N\setminus Q
\]
\rev{with $Q=\{i\in N: \phi_i\neq 0\}$}. The trace of the Jacobian matrix, $tr(JF(\phio))$, is given by
\[
tr(JF(\phio)) = 2r[|Q|-1] - n.
\]
\end{theorem}

\begin{proof}
Consider the trace of the Jacobian at $\phio$, i.e.,
\[
tr(JF(\phio))=\sum_{i=1}^n\dfrac{\partial F_i(\phio)}{\partial \phi_{i}}
\]
Observe that, for $k\neq i$, $\frac{\partial \phio_k}{\partial \phio_i}=-1$, since $\sum_j \phio_j=1$ and thus $\phio_k=1-\sum_{j\neq k}\phio_j$. We have
	 \begin{align*}
	 \dfrac{\partial F_i(\phio)}{\partial \phi_{i}}&= \frac{\partial}{\partial \phi_i}\left[ \frac{\oq_if(\phi_i)}{\sum_k \oq_kf(\phi_k)}-\phi_i\right](\phio)\\
	 &= \frac{\oq_if'(\phio_i)}{\sum_k \oq_kf(\phio_k)}- \underbrace{\frac{\oq_if(\phio_i)}{\sum_k \oq_kf(\phio_k)}}_{p_i(\phio)}\frac{\oq_if'(\phio_i)-\sum_{k\neq i}\oq_kf'(\phio_k)}{\sum_k \oq_kf(\phio_k)}-1\\
	 &= \frac{1}{\sum_k \oq_kf(\phio_k)}\left[ \oq_if'(\phio_i)+\phio_i\left(-\oq_if'(\phio_i)+\sum_{k\neq i}\oq_kf'(\phio_k)\right)\right]-1\\
	 &=\frac{1}{\sum_k \oq_kf(\phio_k)}\left[  (1-\phio_i)\oq_if'(\phio_i)+\phio_i\left(\sum_{k\neq i}\oq_kf'(\phio_k)\right)\right]-1
	 \end{align*}

 \noindent where we used that $p_i(\phio)=\phio_i$ (since $\phio$ is
 an equilibrium) to move from the second to the third equality. Now,
 when $f(x)=x^r$ for $r>1$, $f'(x)=rx^{r-1}$ and we have
\begin{equation}
\label{deriv}	 \dfrac{\partial F_i(\phio)}{\partial \phi_{i}}=\frac{1}{\sum_k \oq_k(\phio_k)^r}\left[  (1-\phio_i)\oq_ir(\phio_i)^{r-1}+\phio_i\left(\sum_{k\neq i}\oq_kr(\phio_k)^{r-1}\right)\right]-1.
\end{equation}
	
\noindent If $i\in N\setminus Q$, then $\phio_i=0$ and $\dfrac{\partial F_i(\phio)}{\partial
  \phi_{i}}=-1$. If $i\in Q$, it follows that
 \begin{align*}
  \dfrac{\partial F_i(\phio)}{\partial \phi_{i}} &=\frac{r}{\sum_{k\in Q}\oq_k\left(\frac{\oq_k^{\frac{1}{1-r}}}{\sum_{j\in Q }\oq_j^{\frac{1}{1-r}}}\right)^r}  \left[(1-\phio_i)\oq_i\left(\frac{\oq_i^{\frac{1}{1-r}}}{\sum_{j\in Q}\oq_j^{\frac{1}{1-r}}}\right)^{r-1}+\phio_i\left(\sum_{k\in Q\setminus\{i\}}\oq_k\left(\frac{\oq_k^{\frac{1}{1-r}}}{\sum_{j\in Q}\oq_j^{\frac{1}{1-r}}}\right)^{r-1}\right)\right]-1\\
	&=\frac{r}{\sum_{k\in Q}\oq_k\frac{\oq_k^{\frac{r}{1-r}}}{\sum_{j\in Q}\oq_j^{\frac{1}{1-r}}}}  \left[(1-\phio_i)\oq_i\oq_i^{-1}+\phio_i\left(\sum_{k\in Q\setminus\{i\}}\oq_k\oq_k^{-1}\right)\right]-1.\\
	 \end{align*}
Since $\sum_{k\in Q}\oq_k\frac{\oq_k^{\frac{r}{1-r}}}{\sum_{j\in Q}\oq_j^{\frac{1}{1-r}}}=1$, we have
$
\dfrac{\partial F_i(\phio)}{\partial \phi_{i}} =r\left[1-\phio_i+\phio_i(|Q|-1) \right]-1=r[1+(|Q|-2)\phio_i]-1.
$
As a result, the trace of the Jacobian at $\phio$ is given by
\begin{align*}
tr(JF(\phio))=\sum_{i=1}^n\dfrac{\partial F_i(\phio)}{\partial \phi_{i}}&=\sum_{i\in Q}(r[1+(|Q|-2)\phio_i]-1)+\sum_{i\in N\setminus Q}(-1)\\
 &=r[|Q|+(|Q|-2)\sum_{i\in Q}\phio_i]-|Q|-(|N|-|Q|)\\
 &=2r[|Q|-1]-n.
 \end{align*}
\end{proof}

\begin{corollary}
Under a social signal $f(x)=x^r, r>1$, the inner equilibrium $\phi^*$ is unstable.
\end{corollary}
\begin{proof}
By Theorem \eqref{theorem:trace}, we have that $tr(JF(\phi^*)) =
2r[|Q|-1]-n$. Since $\phi^*$ has $n$ non-zero market shares, it
follows that $tr(JF(\phi^*)) = 2r[n-1]-n = (r-1)n + r(n-2) > 0$, since
$r > 1$ and $n\geq 2$. As a result, there exists an eigenvalue
$\lambda=\lambda(\phio)$ satisfying $Re(\lambda)>0$.  By Theorems
\eqref{ODEunst} and \eqref{trace}, $\phi^*$ is unstable.
\end{proof}

\begin{remark}
Theorem \eqref{theorem:trace} can also be used to show that many other
equilibria are unstable: They simply need to have enough non-zero
market shares to satisfy $2r[|Q|-1]>n$. Moreover, the theorem can also
be used to show that, for any equilibrium $\phi$ that is not a
monopoly, there exists $r>1$ that makes $\phi$ unstable. It suffices
to choose $r>\frac{n}{2(|Q|-1)}$. For instance, for $n=4$, all the
equilibria but the monopolies are unstable as soon as $r>2$.
\end{remark}

%

\section{Agent-Based Simulation Results}
\label{section:simulation}

We now report results from an agent-based simulation to highlight and
complement the theoretical analysis. The agent-based simulation uses
the setting from \citep{Rank}, which used a dataset to emulate an
environment similar to the \musiclab{}. The setting consists of 50
songs with the values of qualities and appeals specified in Appendix
\ref{dataset}. As mentioned in the introduction, the \musiclab{} is a
trial-offer market where participants can try a song and then decide
to download it. The generative model of the \musiclab{}
\citep{krumme2012quantifying} uses the consumer choice preferences
described in Section \ref{section:market}.

From this section and onwards, we assume that, it each period, a new
customer arrives and may or may not buy a product based on the
probability (quality) of the product tried. (Note that, in the earlier
sections, each new period began when a product was purchased). The reason for this change is our interest in quantifying
the expected number of purchases per period, and how it changes
depending on different ranking policies. We use the expected number of
purchases per period as way to measure the \emph{market
  efficiency}. This view obviously does not change any result from the
previous sections.

\paragraph{The Simulation Setting}

\begin{figure}[!t]
\begin{centering}
\includegraphics[width=0.5\linewidth]{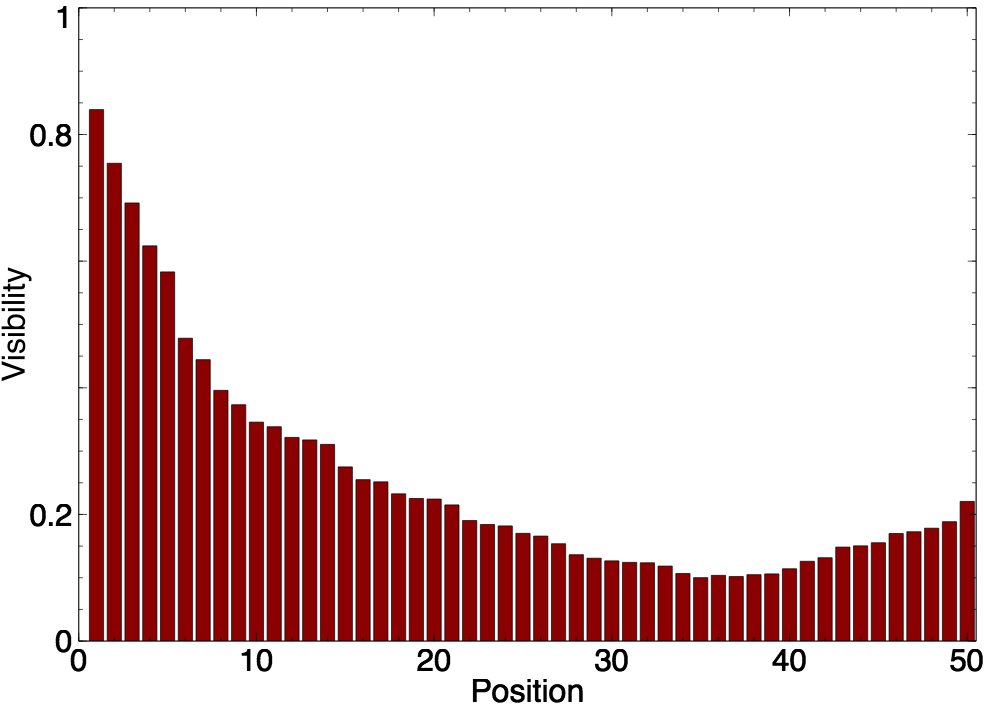}
\end{centering}
\centering{}
\caption{\small The visibility $v_p$ (y-axis) of position $p$ in the
  song list (x-axis) where $p=1$ is the top position and $p=50$ is
  the bottom position of the list which is displayed in a single column.}
\label{fig:visibility}
\end{figure}

\begin{figure}[!t]
\begin{centering}
\includegraphics[width=0.5\linewidth]{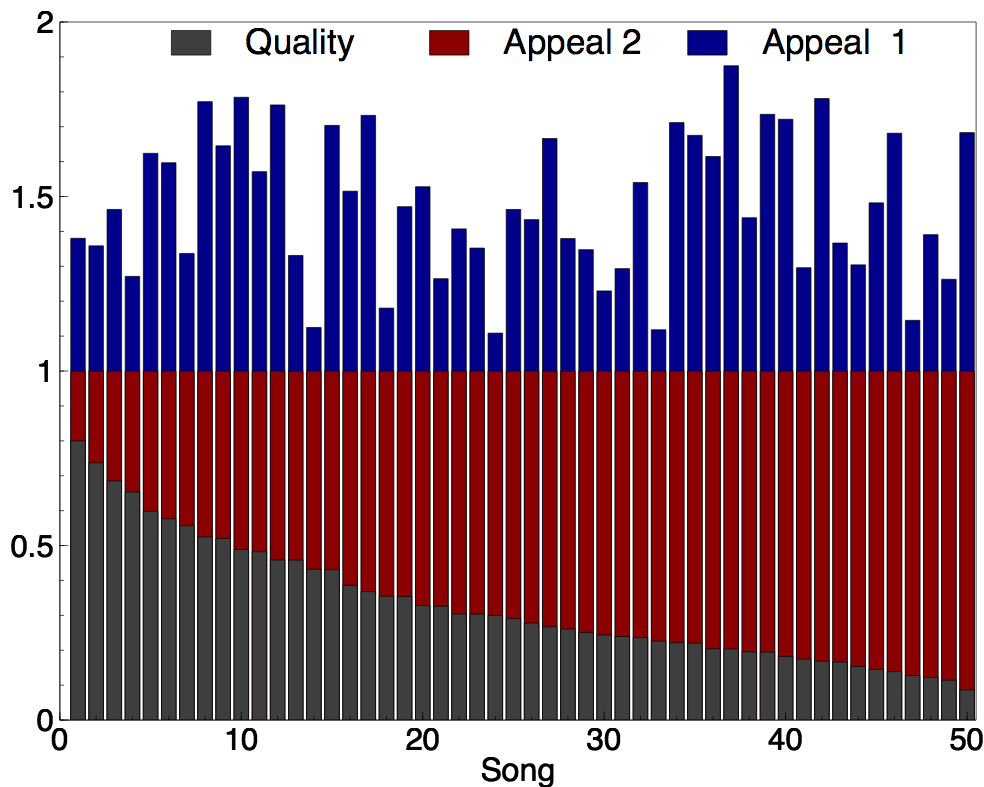}
\end{centering}
\centering{}
\caption{\small The \rev{quality} $q_i$ (grey) and \rev{appeal} $a_i$ (red and blue) of
  song $i$ in the two settings. The settings only differ in the appeal
  of songs, and not in the quality of songs. In the first setting, the
  quality and the appeal for the songs were chosen independently
  according to a Gaussian distribution normalised to fit between 0 and
  1. The second setting explores an extreme case where the appeal is
  anti-correlated with the quality used in setting 1. In this
  second setting, the appeal and quality of each song sum to 1.}
\label{fig:qualityAppeal}
\end{figure}

The agent-based simulation aims at emulating the \musiclab{}: Each
simulation consists of $L$ iterations ($L$ simulated users ) and, at each iteration $t: 0<t<L$,

\begin{enumerate}
\item the simulator randomly selects a song $i$ according to the
  probabilities $P_i(\sigma,\phi)$, where $\sigma$ is the ranking
  proposed by the policy under evaluation and $\phi$ represents the
  market shares;

\item the simulator randomly determines, with probability $q_i$,
  whether selected song $i$ is downloaded. In the case of a download,
  the simulator increases the number of downloads of song $i$, i.e.,
  $d_i^{t+1} = d_i^t + 1$, changing the market shares. Otherwise,
  $d_i^{t+1} = d_i^t$.
\end{enumerate}

\noindent
Every $t>0$ iterations, a new list $\sigma$ may be recomputed if the
ranking policy is dynamic (e.g., the popularity ranking).  In this
paper, the simulation setting focuses mostly on two policies for
ranking the songs:
\begin{itemize}
\item The {\em quality ranking} (Q-rank) that assigns the songs in decreasing order
  of quality to the positions in decreasing order of visibility (i.e.,
  the highest quality song is assigned to the position with the highest visibility
  and so on);

\item The {\em popularity ranking} (D-rank) that assigns the songs in
  decreasing order of popularity (i.e., $d_i^t$) to the positions in
  decreasing order of visibility (i.e., the most popular song is
  assigned to the position with the highest visibility and so on);
\end{itemize}

\noindent
Note that the popularity ranking was used in the original \musiclab{},
while the quality ranking is a static policy: the ranking remains the
same for the entire simulation. The simulation setting, which aims at
being close to the \musiclab{} experiments, considers 50 songs and
simulations with L=$10^5$ iterations unless stated otherwise. The songs are
displayed in a single column. The analysis in
\citep{krumme2012quantifying} indicated that participants are more
likely to try songs higher in the list. More precisely, the visibility
decreases with the list position, except for a slight increase at the
bottom positions. Figure \ref{fig:visibility} depicts the visibility
profile based on these guidelines, which is used in all computational
experiments. The paper also uses two settings for the quality and
appeal of each song, which are depicted in Figure
\ref{fig:qualityAppeal}. In the first setting, the quality and the
appeal were chosen independently according to a Gaussian distribution
normalised to fit between $0$ and $1$. The second setting explores an
extreme case where the appeal is anti-correlated with quality: The
quality is the same as in the first setting but the appeal is chosen
such that the sum of appeal and quality is 1. 

\subsection{Convergence}

We first illustrate the convergence of the market for various
popularity signals ($r < 1$) using the quality ranking. In order to
visualise the results, we focus on only 5 songs, where the qualities,
appeals, and visibilities are given by
\begin{tabbing}
q = [ 0.80, 0.72, 0.68, 0.65, 0.60 ] \\
a = [ 0.38, 0.35, 0.46, 0.27, 0.62 ] \\
v = [ 0.80, 0.75, 0.69, 0.62, 0.58 ].
\end{tabbing}

\noindent
The simulation is run for $10^5$ iterations for the social signals
$f(x)=x^r (r\in \{0.1, 0.25, 0.5,0.75\})$ and Figure
\ref{fig:sublinfeed} depicts the simulation results. Observe that the
equilibrium $\phi^*$ (dashed lines) changes because it depends of the
value of $r$. Interestingly, for social signals with $r\leq 0.5$, the
convergence of the process seems to occur around $10^4$ time steps (iterations)
even when they start with a strong initial distortion due to the
appeals of the songs. The simulations show clear differences in
\rev{behaviour} depending on $r$ and, when $r$ moves closer to 1, the market
tends to exhibit a monopolistic \rev{behaviour} for the song with the best
quality (confirming the results obtained in \citep{ICWSM16SI}).

\begin{figure}[h]
\vspace{-5mm}
\begin{centering}
\includegraphics[scale=0.37]{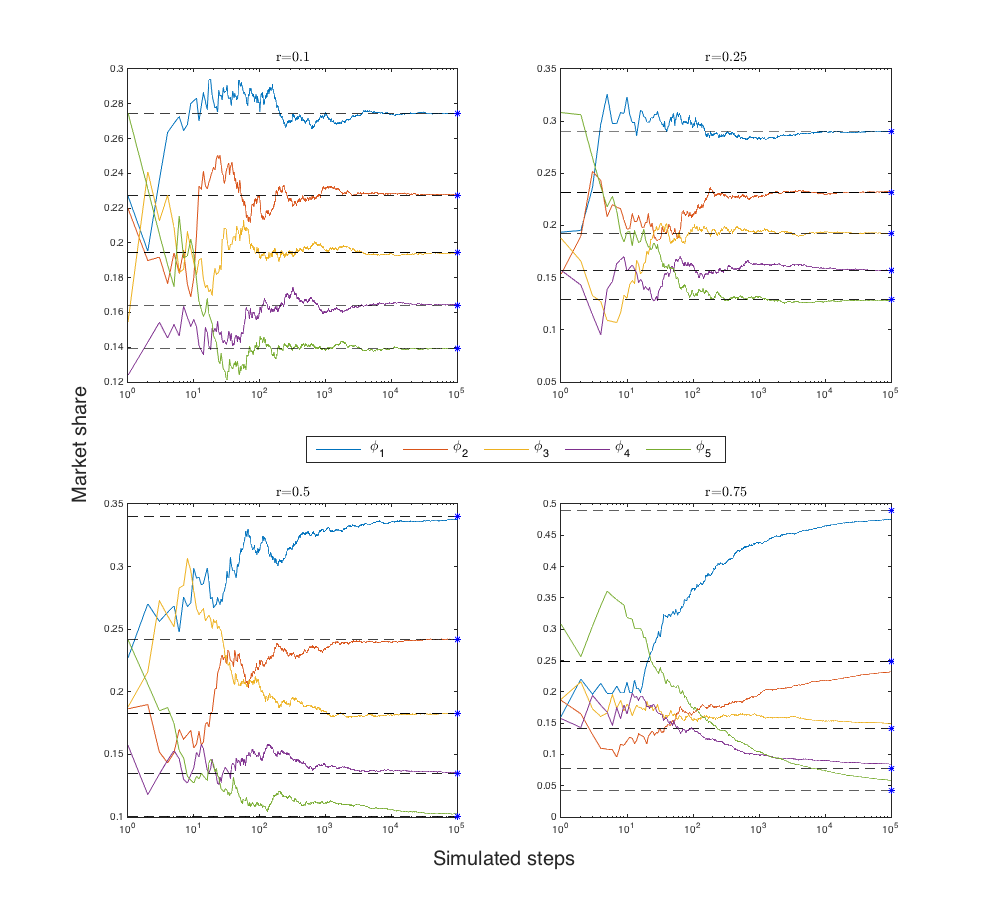}
\end{centering}
\centering{}
\vspace{-2mm}
\caption{\small Evolution of market shares of 5 songs using a social signal $f(x)=x^r,\quad r\in \{0.1, 0.25, 0.5,0.75\}$. Dashed lines are the values of the equilibrium for each song.}
\label{fig:sublinfeed}
\end{figure}

\begin{figure}[!t]
\vspace{-3.5mm}
\begin{centering}
\includegraphics[scale=0.31]{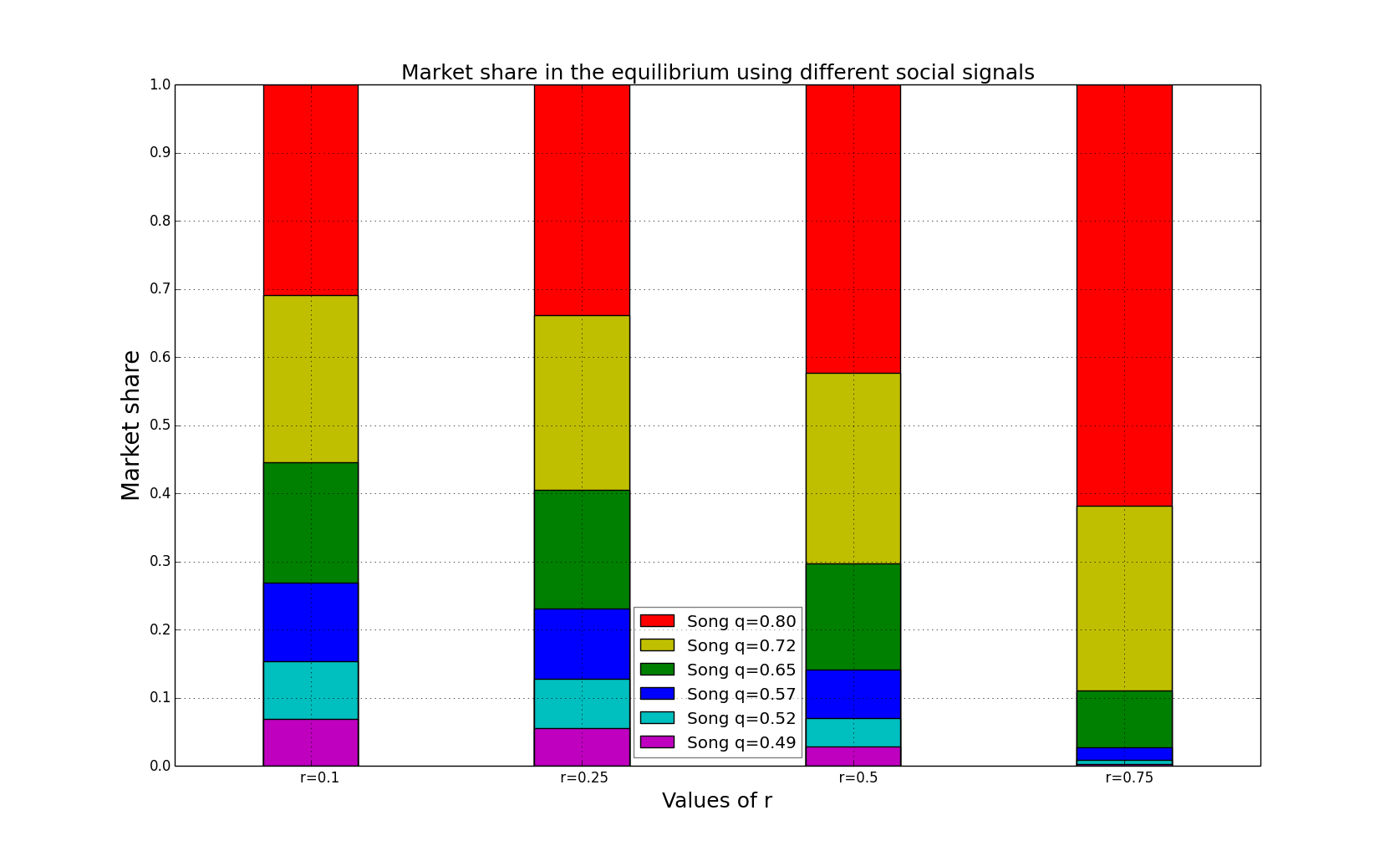}
\end{centering}
\centering{}
\vspace{-3.5mm}
\caption{\small Market shares of 6 songs and their qualities, using a social signal $f(x)=x^r,\quad r\in \{0.1, 0.25, 0.5,0.75\}$.}
\label{fig:fairness}
\vspace{-2.5mm}
\end{figure}

Figure \ref{fig:fairness} shows how the market is distributed in the
equilibrium among 6 songs. The qualities,
appeals, and visibilities are given by
\begin{tabbing}
 q = [0.80, 0.72, 0.65, 0.57, 0.52, 0.49] \\
 a = [0.38, 0.36, 0.27, 0.60, 0.77, 0.78] \\
 v = [0.80, 0.75, 0.62, 0.48, 0.40, 0.35]
\end{tabbing}
\noindent
and the social signals are of the form $f(x)=x^r$ $(r\in \{0.1, 0.25,
0.5, 0.75\})$. Each stacked bar represents the proportion of the
market for the 6 songs for a given social signal. Songs with better
qualities (i.e., the top 2 songs represented in red and yellow
respectively) have larger market shares and their market shares
increase with $r$. In contrast, the market shares of the lower-quality
songs (i.e., cyan and purple respectively) decrease when $r$
increases. These results indicate that social influence has a
beneficial effect on the market: it drives customers towards the
better products, while not going to a monopoly as long as $r < 1$.

\subsection{Market Predictability}

\begin{figure}[!t]
\vspace{-5mm}
\centering
\includegraphics[scale=0.40]{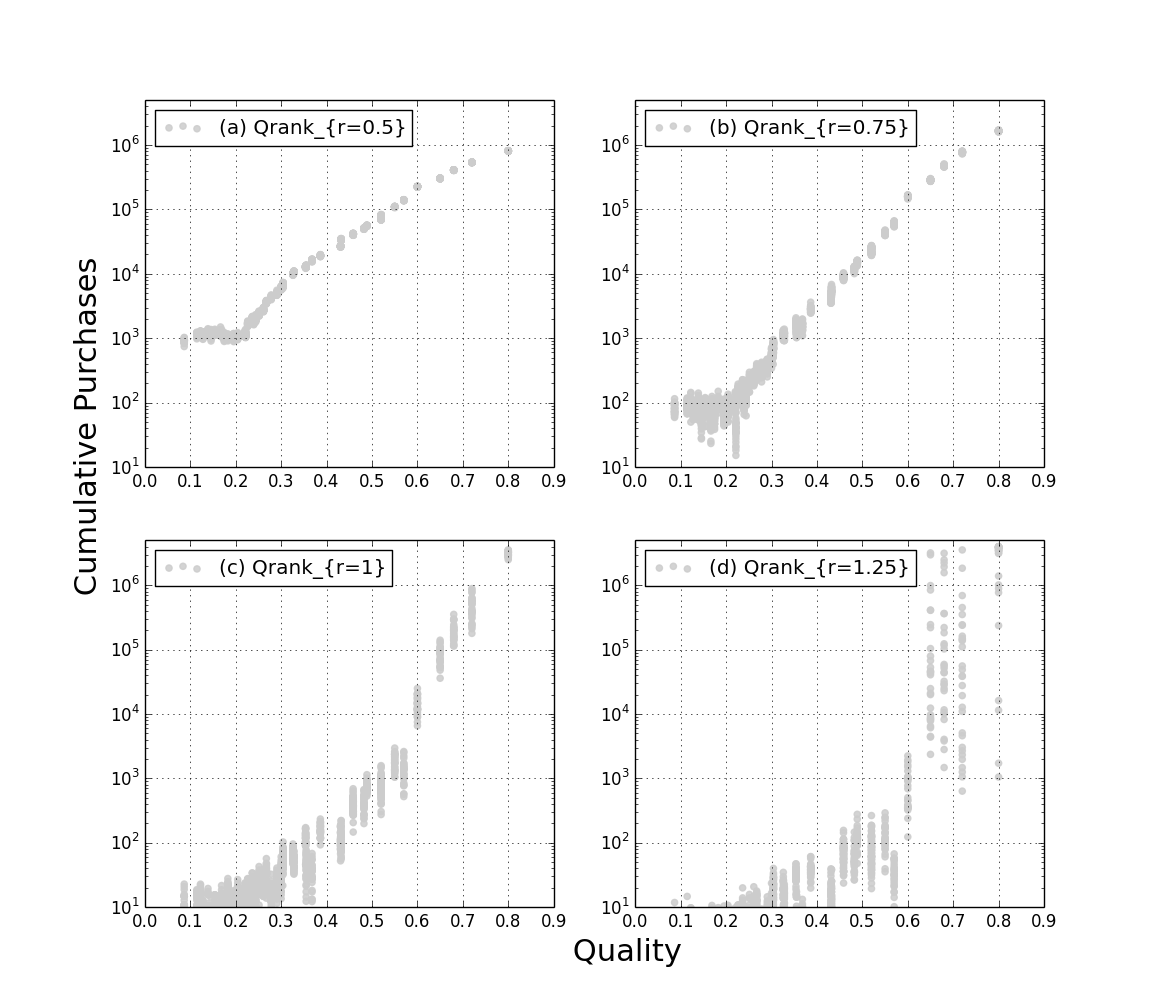}
\centering{}
\vspace{-3mm}
\caption{\small Distribution of downloads versus the qualities, using social signals $f(x)=x^r,\quad r\in \{0.5, 0.75, 1, 1.25\}$. The results are for the first setting where the quality and appeal of each song are chosen independently. The songs are ordered by increasing quality along the x-axis. The y-axis is the number of downloads.}
\label{fig:predic1}
\end{figure}

\begin{figure}[!t]
\vspace{-5mm}
\centering
\includegraphics[scale=0.40]{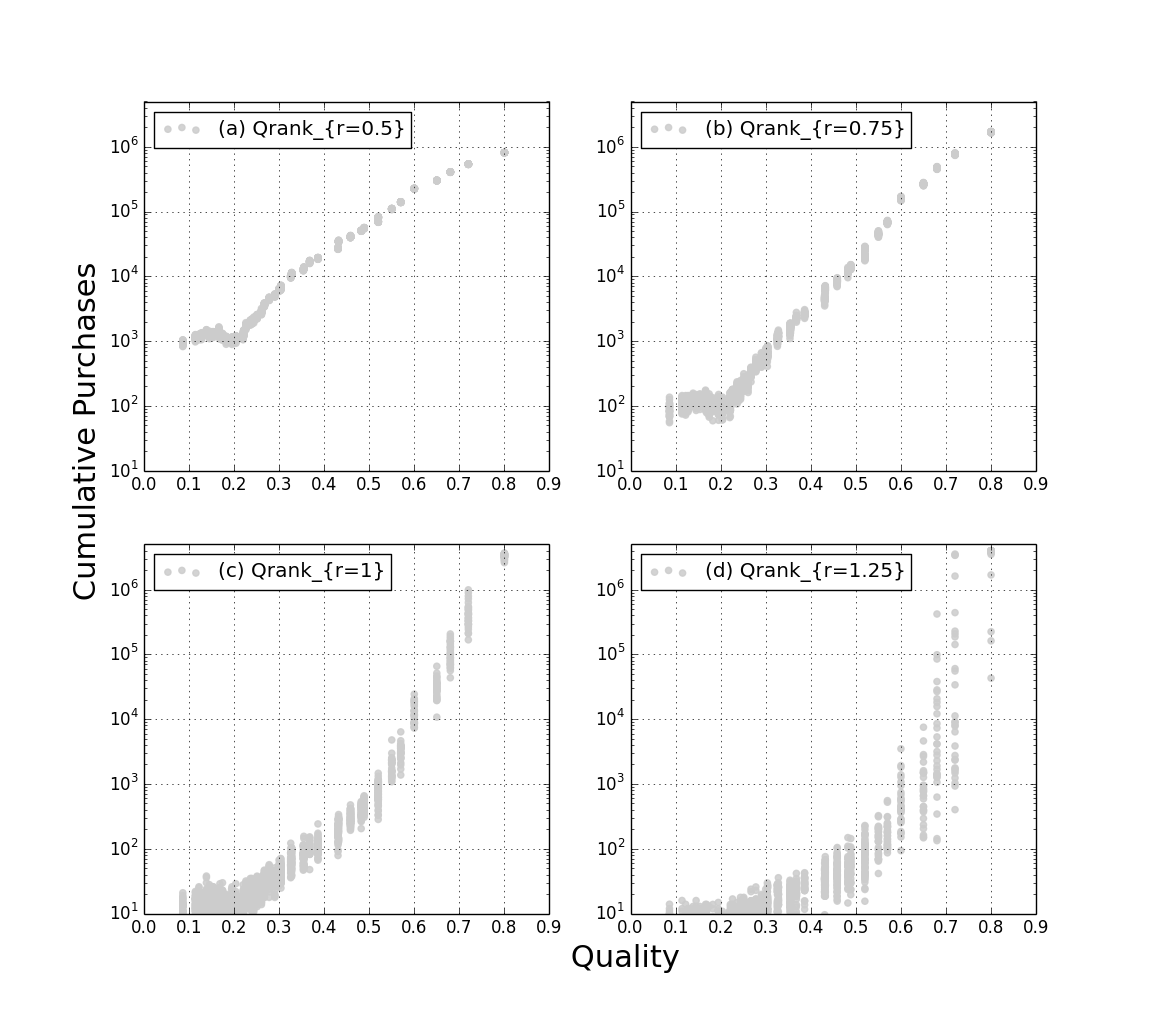}
\centering{}
\vspace{-3mm}
\caption{\small Distribution of downloads versus the qualities, using social signals $f(x)=x^r,\quad r\in \{0.5, 0.75, 1, 1.25\}$. The results are for the first setting where the quality and appeal of each song are anti-correlated. The songs are ordered by increasing quality along the x-axis. The y-axis is the number of downloads.}
\vspace{-1mm}
\label{fig:predic2}
\end{figure}

This section depicts the predictability of the market for various
values of $r$ and the number of downloads per song as a function of
its quality. Figures \ref{fig:predic1} and \ref{fig:predic2} depict
the results for the two quality/appeal settings discussed previously.
The figures display the results of 40 experiments for each setting
with 1 million arrivals. Each experiment contributes 50 data points,
i.e., the number of downloads for each song, and all the data points
for the 40 experiments are displayed in the figures.

In the plots, the
x-axis represents the song qualities and the y-axis the number of
downloads.  A dot at location $(q,d)$ indicates that the song with
quality $q$ had $d$ downloads in an experiment. Obviously, there can be
several dots at the same location. For $r \in \{0.5,0.75,1\}$, the
market is highly predictable and there is little variation in the song
downloads. For $r=1$, the market converges to a monopoly for the song
of highest quality, confirming the results from
\citep{Rank,ICWSM16SI}. Finally, for $r=1.25$, the market exhibits
significant unpredictability, as suggested by the theoretical
results. In this case, the equilibria are monopolies for various songs
but it is hard to predict which song will dominate the market.

Note
also that the unpredictability of the market increases significantly
for $r=1.25$ when the appeal and quality of the songs are
anti-correlated. This is not the case for $r \in \{0.5,0.75\}$. To evaluate the statistical significance of these results, we measure
the market unpredictability as suggested by Salganik et
al. \cite{salganik2006experimental}.  The unpredictability $u_i$ for
product $i$ is defined as the average difference in market share for
that product over the 40 experiments:
\[
u_i=\frac{1}{\left( \begin{array}{c} 40 \\ 2 \end{array} \right)} \sum_{w=1}^{40}\sum_{w^*=w+1}^{40}|\phi_{i,w}-\phi_{i,w^*}|,
\]
where $\phi_{i,w}$ is the final market share of product $i$ in
experiment $w$. We then computed the overall unpredictability for each social signal $r\in \{0.5, 0.75, 1, 1.25\}$: $U=\dfrac{\sum_{j=1}^nu_j}{n}.$

Figure \ref{fig:unpreIND} shows the average unpredictability $U$ and the standard deviation for the different social signals, using the
same data as in Figures \ref{fig:predic1} and \ref{fig:predic2}
(Figure \ref{fig:unpreIND} $\mathsf{a}$ and Figure
\ref{fig:unpreIND} $\mathsf{b}$ respectively). We also performed
Mann-Whitney U tests, comparing the values of $U$ for pairs of social
signals.  In all cases, a social signal $r< 1$ is significantly
more predictable than the signal $r=1.25$ ($p$-value<0.05). Comparisons
between $r=0.5$ and $r=0.75$ and $r=0.75$ and $r=1$ also show
statistically significant differences in unpredictability. For instance,
for the anti-correlated setting, the $p$-values for the various pairwise comparisons ({\it first column} is less unpredictable than {\it second column})
are given in Table \ref{pvalue1}
\begin{table}[h]
\vspace{2mm}
\centering
\begin{tabular}{lll}
\hline
social signal & social signal & $p$-value \\
\hline
0.5 & 0.75 &0.0029 \\
0.5 & 1  &  8.73e-07 \\
0.5 & 1.25 & 4.24e-10 \\
0.75 & 1 & 0.0003 \\
0.75 & 1.25 & 2.10e-07 \\
1 & 1.25  & 0.0022  \\
\hline \\
\end{tabular}
\vspace{-3mm}
\caption{\small p-values of the  hypothesis: first column is less predictable than second column. Case  $a_i,q_i$ anti-correlated.}
\label{pvalue1}
\vspace{2mm}
\end{table}

\noindent
For the independent setting, Table \ref{pvalue2} shows the pairwise comparisons that are also statistically significant in that case:
\begin{table}[!h]
\vspace{-1mm}
\centering
\begin{tabular}{lll}
\hline
social signal & social signal & p-value \\
\hline
0.5 & 1 & 0.0414 \\
0.5 & 1.25 & 0.0034 \\
0.75 & 1.25 & 0.0266 \\
\hline \\
\end{tabular}
\vspace{-4mm}
\caption{\small p-values of the  hypothesis: first column is less predictable than second column. Case $a_i,q_i$ independent.}
\label{pvalue2}
\end{table}

\begin{figure}[!t]
\vspace{-3mm}
\centering
\includegraphics[scale=0.47]{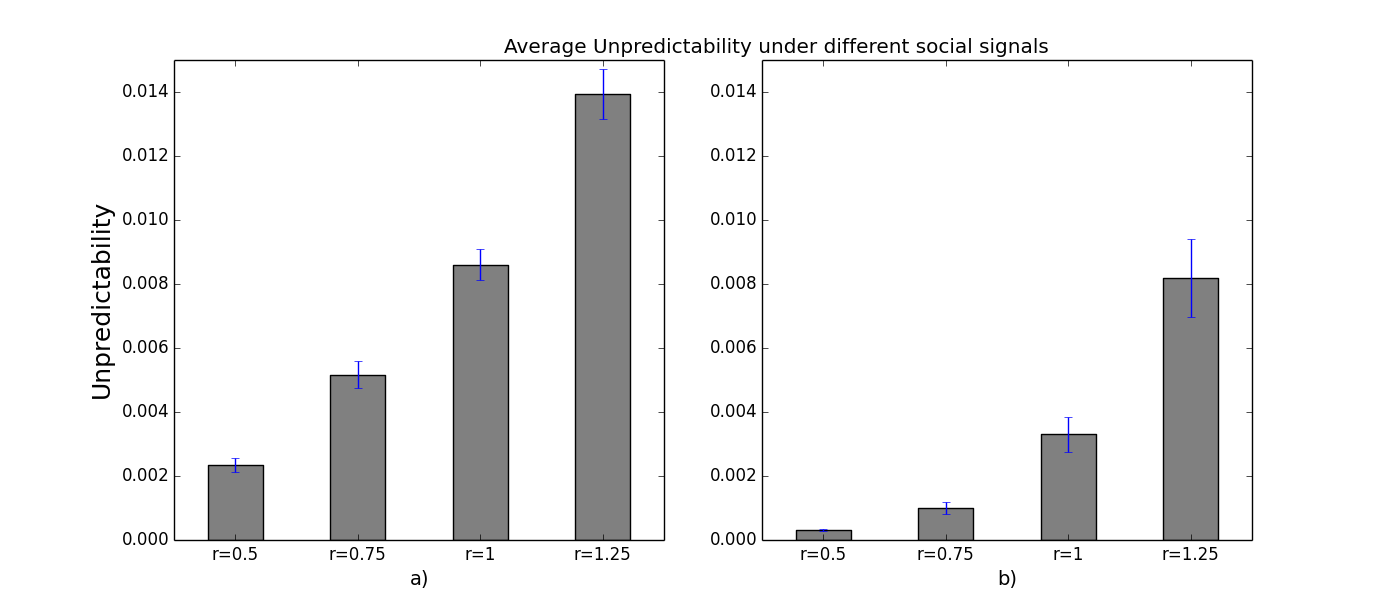}
\centering{}
\vspace{-5mm}
\caption{\small Average unpredictability (grey bars), using social signals $f(x)=x^r,\quad r\in \{0.5, 0.75, 1, 25\}$. $\mathsf{a})$ shows the results for the independent setting, and $\mathsf{b})$ for the anti-correlated setting. Both cases consist of 40 experiments with 1 million iterations each. Blue lines represent the respective standard deviations.}
\label{fig:unpreIND}
\end{figure}

\noindent Figure \ref{fig:predicqd} compares the predictability of Q-rank and
D-rank for the first setting of Quality/Appeal. For each ranking, two
different social signals were used ($r=0.5$ and $r=1$) and the figure
displays the result of 50 experiments, consisting in 1 million iterations. Two phenomena can be observed. First, sublinear signals seem to
help the D-rank, making the outcome less chaotic (first
column). Second, Q-rank clearly performs better than D-rank and
exhibits much less unpredictability.


\begin{figure}[!t]
\vspace{-5mm}
\begin{centering}
\includegraphics[scale=0.42]{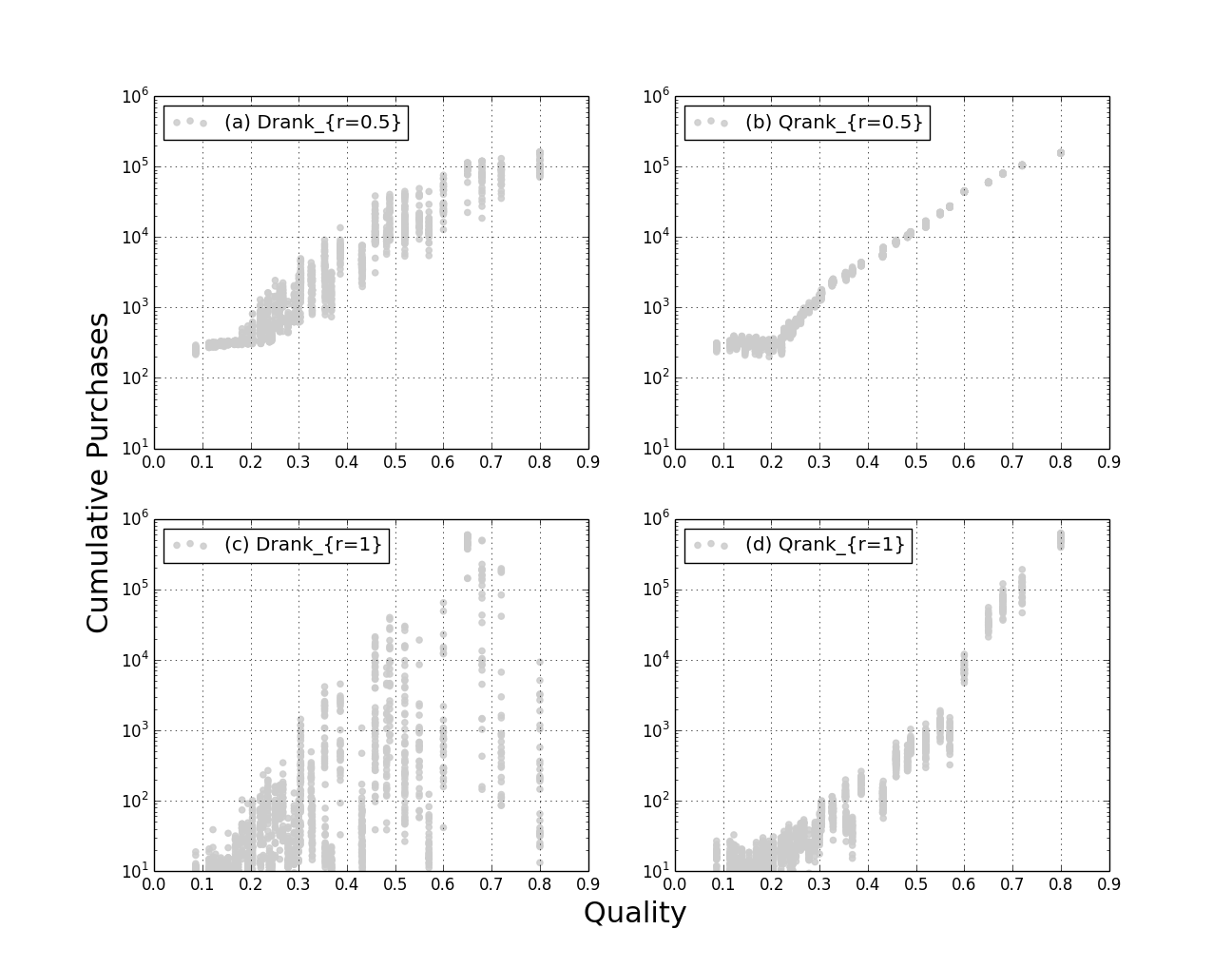}
\end{centering}
\centering{}
\vspace{-3mm}
\caption{\small Distribution of purchases versus product qualities for
  50 experiments with 1 million users. Figures (a) and (b) use a
  social signal $f(x)=x^{0.5}$, Figure (a) shows the results for the
  popularity ranking and Figure (b) for the quality ranking.  Figures
  (c) and (d) use the social signal $f(x)=x$, Figure (c) shows the
  results for the popularity ranking and Figure (d) for the quality
  ranking.}
\label{fig:predicqd}
\end{figure}

\subsection{Performance of the Market}

Figures \ref{fig:ranki1} and \ref{fig:ranki2} report results about the
performance of the markets as a function of the social influence
signals. The figures report the average number of downloads over time among 50 experiments,
for the quality and popularity rankings as a function of the social
signals. There are a few observations that deserve mention.

\begin{enumerate}
\item For the quality ranking, the expected number of downloads
  increases with the strength of the social signal as $r$ approaches
  1. The equilibrium when $r=1$ is optimal asymptotically and assigns
  the entire market share to the song of highest quality. When $r=2$,
  the situation is more complicated. The figure shows that the market
  efficiency can further improve if $r=2$. However, when the
  simulation is run for more iterations (a result not shown in the
  figure), the market efficiency decreases slightly compared to $r=1$,
  which is consistent with the theory since there is no guarantee that
  the monopoly for $r>1$ is for the song of highest quality.


\item The popularity ranking is always dominated by the quality
  ranking and the benefits of the quality ranking increase as $r$
  approaches 1 from below.

\item The popularity ranking in the second setting when $r=2$, obtains
  nearly a third of the expected downloads than the quality ranking.
\end{enumerate}

\begin{figure}[!t]
\vspace{-5mm}
\centering
\includegraphics[scale=0.55]{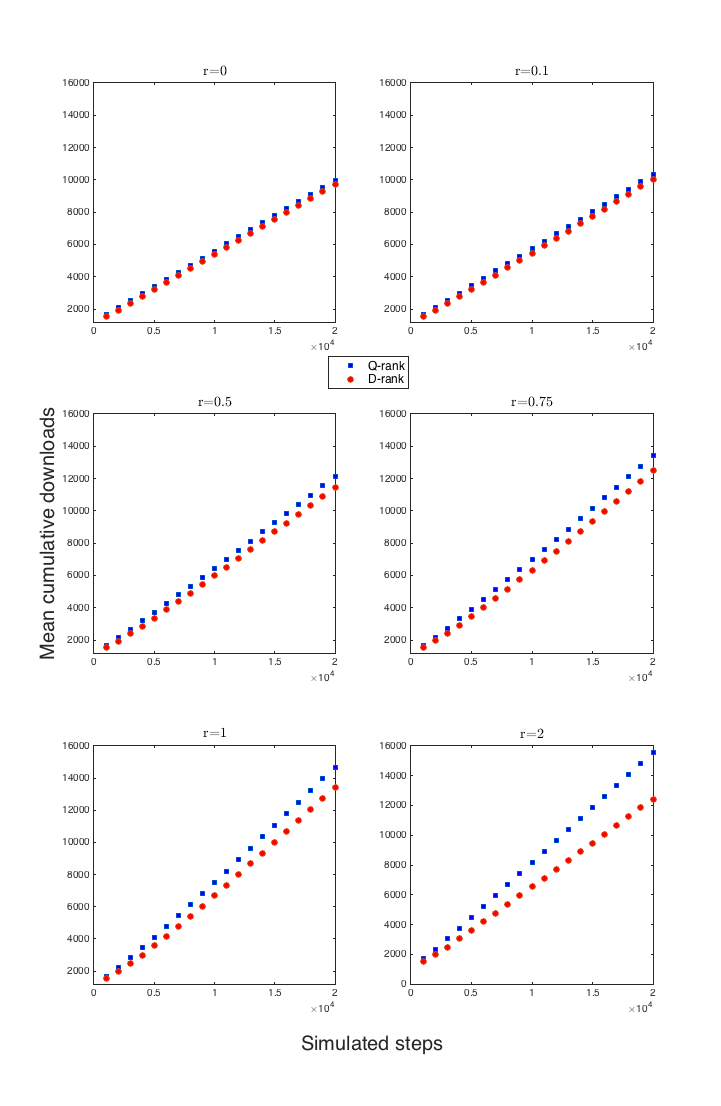}
\vspace{-1.2cm}
\caption{\small The Average Number of Downloads over Time for the Quality and Popularity Rankings for Various Social Signals in the First Setting for Song Appeal and Quality.}
\label{fig:ranki1}
\end{figure}

\begin{figure}[!t]
\centering
\includegraphics[scale=0.52]{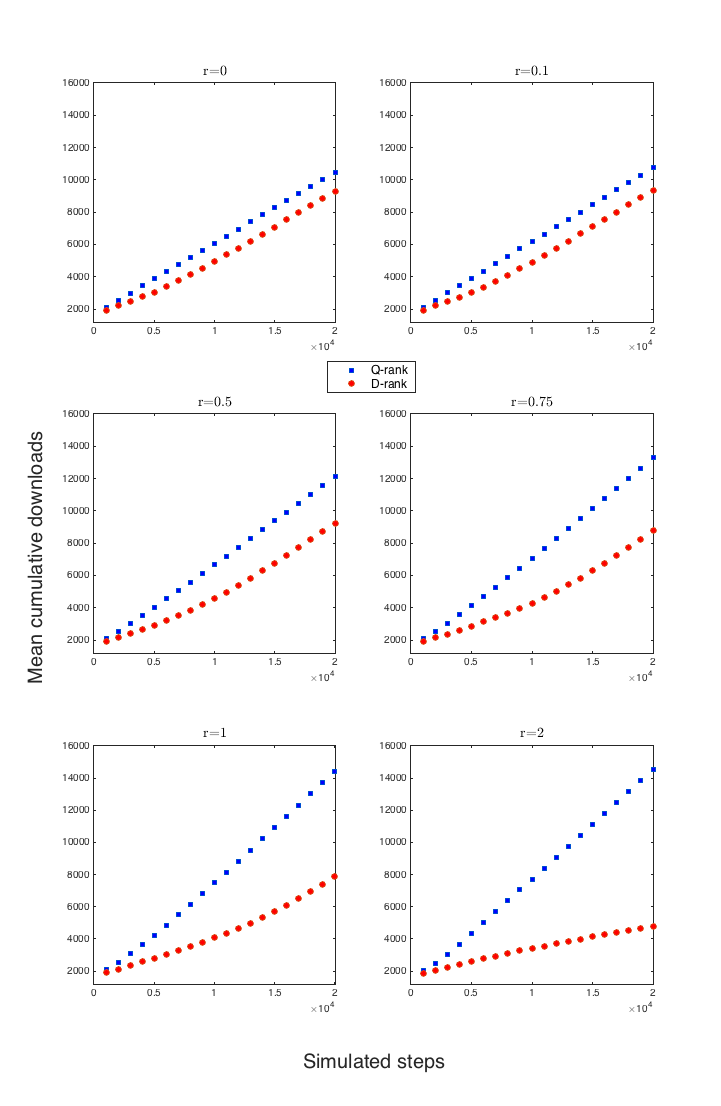}
\vspace{-0.5cm}
\caption{\small The Average Number of Downloads over Time for the Quality and Popularity Rankings for Various Social Signals in the Second Setting for Song Appeal and Quality.}
\label{fig:ranki2}
\end{figure}

\section{Additional Observations on  Sublinear Social Signals}
\label{section:additional}

\paragraph{The Benefits of Social Influence}

A linear social signal has been shown to be beneficial to the market
efficiency, i.e. it maximises the expected number of downloads. This result was
proved by \cite{Rank} for the performance ranking and by
\cite{ICWSM16SI} for any static ranking such as the quality
ranking. Unfortunately, sublinear social signals are not always
beneficial to the market in that sense, as one can see in Example \ref{example_social_signal_sublinear}.
Consider, once again, the quality ranking and assume that $q_1 \geq \ldots \geq q_n$. When there is no social
signal, by following the idea in Equation \eqref{probkrum} and taking
$f(\phi_i)=\beta \phi_i+\alpha a_i$, with $\beta=0$, the probability
of trying product $i$ is given by
\[
p_i^I=\frac{v_{i}a_i}{\sum_{j=1}^nv_{j}a_j}.
\]
In this case, the expected number of purchases per period is
\[
\sum_{i=1}^n p_i^I q_i.
\]
On the other hand, with a social signal, the probability of
trying product $i$ at time $t$ is
\[
P_i(\phi^t) =  \frac{v_{i} f(\phi_i^t)}{\sum_{j=1}^n v_{j} f(\phi_j^t)}
\]
and the expected number of purchases per period at the equilibrium is given by
\[
\sum_{i=1}^nP_i(\phi^*)q_i =  \sum_{i=1}^n\frac{v_{i} q_if(\phi_i^*)}{\sum_{j=1}^n v_{j} f(\phi_j^*)}.
\]
The following example shows that, under a sublinear social signal, the
expected number of purchases (per period) at equilibrium, i.e.,
$\sum_{i=1}^nP_i(\phi^*)q_i$, can be lower than the expected number of
purchases when no social signal is used, i.e., $\sum_{i=1}^n p_i^I
q_i.$
\begin{example} \label{example_social_signal_sublinear}
  Consider a 2-dimensional T-O market with social signal
  $f(x)=x^{0.5}$, where the qualities, visibilities, and appeals are
  given by
\begin{itemize}
\item $q_1=1$, $q_2=0.4$,
\item $v_1=1$, $v_2=1$,
\item $a_1=1$, $a_2=0.3$.
\end{itemize}
The expected number of purchases at equilibrium for the case with
social signal is given by
\[ \frac{v_1q_1(\phi_1^*)^r+v_2q_2(\phi_2^*)^r}{v_1(\phi_1^*)^r+v_2(\phi_2^*)^r}=\frac{v_1q_1(v_1q_1)^{r/(1-r)}+v_2q_2(v_2q_2)^{r/(1-r)}}{v_1(v_1q_1)^{r/(1-r)}+v_2(v_2q_2)^{r/(1-r)}}=\frac{1+(0.4)^2}{1+0.4}\sim0.83,\]
while, for the case without social signal, it is given by
\[ \frac{v_1q_1a_1+v_2q_2a_2}{v_1a_1+v_2a_2}=\frac{1+0.3(0.4)}{1+0.3}\sim0.86.\]
\end{example}
\noindent

This simple example, in which the qualities and appeals are positively
correlated, shows that if customers follow a sublinear social
influence signal ($r=0.5$), the market efficiency gets reduced by
around 3 percent (with respect to not showing them the social signal).
In contrast, when $r=1$, social influence drives the market towards a
monopoly, which leads to an asymptotically optimal market that assigns
the entire market share to the highest quality product (which may be
undesirable in practice).  Note that, once the qualities and appeals
have been recovered (using, say, Bernoulli sampling as suggested in
\citep{Rank}), one could potentially decide whether to use the social
influence (in case it is sublinear $r < 1$): Simply compare the
expected number of purchases in both settings, using the equilibrium
for the social influence case and the formula for the case with no
social signal.

\paragraph{Optimality of the Quality Ranking}

When $r=1$, it has been shown that the quality ranking is optimal
asymptotically: It maximises the expected number of purchases
\cite{ICWSM16SI}. If another static ordering is used, the market will
converge to the product that has the highest quality when scaled by
its visibility. However, when $0<r<1$, the quality ranking
is no longer guaranteed to be optimal asymptotically.

\begin{example}
Consider a 3-dimensional T-O market with a social signal $f(x)=x^{r},
r=0.3$, and the following values for qualities and visibilities:
\begin{itemize}
\item $q_1=1$, $q_2=0.261$, $q_3=0.002$
\item $v_1=1$, $v_2=0.720$, $v_3=0.229$
\end{itemize}
then, using quality ranking we would end up with an expected number of purchases at equilibrium, given by
\begin{align*}
\sum_{i=1}^nP_i(\phi^*)q_i&=\frac{v_1q_1(v_1q_1)^{r/(1-r)}+v_2q_2(v_2q_2)^{r/(1-r)}+v_3q_3(v_3q_3)^{r/(1-r)}}{v_1(v_iq_i)^{r/(1-r)}+v_2(v_2q_2)^{r/(1-r)}+v_3(v_3q_3)^{r/(1-r)}}\\
&=\frac{1+(0.720*0.261)^{10/7}+(0.229*0.002)^{10/7}}{1+0.720(0.720*0.261)^{3/7}+0.229(0.229*0.002)^{3/7}}\sim 0.8026,
\end{align*}

on the other hand, if we decide to place the third product (quality $q_3=0.002$)  in the second position, and the second product (quality $q_2=0.261$) in the third position of the ranking,  we get
\begin{align*}
\sum_{i=1}^nP_{\sigma_i}(\phi^*)q_i&=\frac{v_1q_1(v_1q_1)^{r/(1-r)}+v_2q_3(v_2q_3)^{r/(1-r)}+v_3q_2(v_3q_2)^{r/(1-r)}}{v_1(v_iq_i)^{r/(1-r)}+v_2(v_2q_3)^{r/(1-r)}+v_3(v_3q_2)^{r/(1-r)}}\\
&=\frac{1+(0.720*0.002)^{10/7}+(0.229*0.261)^{10/7}}{1+0.720(0.720*0.002)^{3/7}+0.229(0.229*0.261)^{3/7}}\sim 0.9154.
\end{align*}
\end{example}

\noindent
The intuition behind the previous example is that if there exists a
product which is much better than the rest, the best decision is to
exhibit it in the first position and place, in the second position,
the lowest quality product to make the first product is even more
appealing. It is an open problem to determine whether there is
a polynomial-time algorithm to find an optimal ranking.

%
%
%

%

\section{Discussion and Conclusion}
\label{section:conclusion}

This paper studied the role of social influence in trial-offer markets
where customer preferences are modeled by a generalisation of a
multinomial logit. In this model, both position bias and social
influence impact the products tried by consumers.

The main result of the paper is to show that trial-offer markets, when
the ranking of the products is fixed, converge to a unique equilibrium
for sublinear social signals of the form $\phi_i^r$, where $\phi_i$
represents the cumulative market share of product $i$. Of particular
interest is the fact that the equilibrium does not depend on the
initial conditions, e.g., the product appeals, but only depends on the
product qualities. Moreover, when the products are ranked by quality,
i.e., the best products are assigned the highest visibilities, the
equilibrium is such that the better products receive the largest
market shares, which increase as $r$ increases for the best products
(as long as $r<1$). The equilibrium for a sublinear social signal
contrasts with the case with $r=1$, where the market goes to a
monopoly for the highest quality product (under the quality
ranking). In the sublinear case, the market shares reflect product
quality but no product becomes a monopoly. The paper also shows that,
when $r>1$, the market becomes more unpredictable. In particular, the
inner equilibrium, which assigns a positive market share to all
products, is unstable and the market is likely to converge to a
monopoly for some product. However, which product becomes the monopoly
depends on the initial conditions.

Simulation results on a setting close to the original \musiclab{}
complemented the theoretical results. They show that the market
converges quickly to the equilibrium for a sublinear social signal and
that the convergence speed depends on the social signal strength. The
simulation results also illustrate how the market shares of the
highest (resp. lowest) quality products increase (resp. decrease) with
$r$. As expected, when $r \leq 1$, the market is shown to be highly
predictable, while it exhibits a lot of randomness when $r>1$. The
simulation results also show the benefits of social influence for
market efficiency, and
demonstrate that the quality ranking once again outperforms the
popularity ranking.

Overall, these results shed a new light on the role of social
influence in trial-offer markets and provide a comprehensive overview
of the choices and tradeoffs available to firms interested in
optimising their markets with social influence. In particular, they
show that social influence does not necessarily make markets
unpredictable and is typically beneficial when the social signal is
not too strong. Moreover, ranking the products by quality appears to
be a much more effective policy than ranking products by popularity
which may induce unpredictability and market inefficiency. The results
also show that sublinear social signals give decision makers the
ability to trade market efficiency for more balanced market shares.

Perhaps, the main contribution of this paper is to show that markets
under social influence are very sensitive to various design
choices. The findings in \cite{salganik2006experimental} used the
popularity ranking, which significantly affected their conclusions
about market unpredictability and efficiency. The theoretical and
simulation results of this paper, together with those in
\citep{Rank,ICWSM16SI} for the case $r=1$, show that the market is
highly predictable when using any static ranking and $r\leq1$.
Moreover, the quality ranking is optimal asymptotically when $r=1$ and
dominates the popularity ranking in all our simulations which were
modeled after the \musiclab{}. This does not diminish the value of the
results by Salganik et al.  \citep{salganik2006experimental} who
isolated potential pathologies linked to social influence. But this
paper shows that these pathologies are not inherent to the market but
are a consequence of specific design choices in the experiment: The
strength of the social signal and the ranking policy. Interestingly,
it is only for a linear social signal that social influence can be
shown to be always beneficial in expectation. Fortunately, for
sublinear social signals, we can determine a priori if social
influence is beneficial, given the analytic form of the equilibrium.

There are at least two potential research directions following this
paper that worth investigating. First, it would be extremely valuable
to construct large-scale cultural market experiment, varying the
strengths of the social signal to complement our simulation
results. Second, it would be interesting to extend our results to
other settings including assortment problems (where the firm can
select not only how to rank products but also which ones should be
shown) \cite{4OR} and to classical cascade models with a social signal
\citep{wang2012}.


\bibliographystyle{abbrv}
\bibliography{references}

\begin{thebibliography}{10}

\bibitem{Rank}
A.~Abeliuk, G.~Berbeglia, M.~Cebrian, and P.~Van~Hentenryck.
\newblock The benefits of social influence in optimized cultural markets.
\newblock {\em PLoS ONE}, 10(4), 2015.

\bibitem{4OR}
A.~Abeliuk, G.~Berbeglia, M.~Cebrian, and P.~Van~Hentenryck.
\newblock Assortment optimization under a multinomial logit model with position
  bias and social influence.
\newblock {\em 4OR}, 14(1):57--75, 2016.

\bibitem{IJCAI16SI}
A.~Abeliuk, G.~Berbeglia, F.~Maldonado, and P.~Van~Hentenryck.
\newblock Asymptotic optimality of myopic optimization in trial-offer markets
  with social influence.
\newblock In {\em In the 25th International Joint Conference on Artificial
  Intelligence (IJCAI-16), New York, NY, July 9-15}, 2016.

\bibitem{abeliuk2017taming}
A.~Abeliuk, G.~Berbeglia, P.~Van~Hentenryck, T.~Hogg, and K.~Lerman.
\newblock Taming the unpredictability of cultural markets with social
  influence.
\newblock In {\em Proceedings of the 26th International Conference on World
  Wide Web}, pages 745--754. International World Wide Web Conferences Steering
  Committee, 2017.

\bibitem{altszyler2017transient}
E.~Altszyler, F.~Berbeglia, G.~Berbeglia, and P.~Van~Hentenryck.
\newblock Transient dynamics in trial-offer markets with social influence:
  Trade-offs between appeal and quality.
\newblock {\em PloS one}, 12(7):e0180040, 2017.

\bibitem{Benaim99}
M.~Bena{\"\i}m.
\newblock Dynamics of stochastic approximation algorithms.
\newblock In {\em Seminaire de probabilites XXXIII}, pages 1--68. Springer,
  1999.

\bibitem{Borkar00}
V.~S. Borkar and S.~P. Meyn.
\newblock The {ODE} method for convergence of stochastic approximation and
  reinforcement learning.
\newblock {\em SIAM Journal on Control and Optimization}, 38(2):447--469, 2000.

\bibitem{Ceyhan11}
S.~Ceyhan, M.~Mousavi, and A.~Saberi.
\newblock Social influence and evolution of market share.
\newblock {\em Internet Mathematics}, 7(2), 2011.

\bibitem{Chung2003}
F.~Chung, S.~Handjani, and D.~Jungreis.
\newblock Generalizations of {P}{\'o}lya's urn problem.
\newblock {\em Annals of Combinatorics}, 7:141--153, 2003.

\bibitem{conley1978isolated}
C.~Conley.
\newblock Isolated invariant sets and the morse index.
\newblock In {\em CBMS Regional Conference Series in Mathematics, 38},
  volume~16, 1978.

\bibitem{craswell2008experimental}
N.~Craswell, O.~Zoeter, M.~Taylor, and B.~Ramsey.
\newblock An experimental comparison of click position-bias models.
\newblock In {\em Proceedings of the 2008 International Conference on Web
  Search and Data Mining}, pages 87--94. ACM, 2008.

\bibitem{Duflo97}
M.~Duflo and S.~S. Wilson.
\newblock {\em Random iterative models}, volume~22.
\newblock Springer Berlin, 1997.

\bibitem{engstrom2014demand}
P.~Engstrom and E.~Forsell.
\newblock Demand effects of consumers' stated and revealed preferences.
\newblock {\em Available at SSRN 2253859}, 2014.

\bibitem{Hirsch2012}
M.~W. Hirsch, S.~Smale, and R.~L. Devaney.
\newblock {\em Differential equations, dynamical systems, and an introduction
  to chaos}.
\newblock Academic press, 2012.

\bibitem{hu2015liking}
M.~Hu, J.~Milner, and J.~Wu.
\newblock Liking and following and the newsvendor: Operations and marketing
  policies under social influence.
\newblock {\em Management Science}, 62(3):867--879, 2015.

\bibitem{jordan1999nonlinear}
D.~W. Jordan and P.~Smith.
\newblock {\em Nonlinear ordinary differential equations: an introduction to
  dynamical systems}, volume~2.
\newblock Oxford University Press, USA, 1999.

\bibitem{krumme2012quantifying}
C.~Krumme, M.~Cebrian, G.~Pickard, and S.~Pentland.
\newblock Quantifying social influence in an online cultural market.
\newblock {\em PLoS ONE}, 7(5), 2012.

\bibitem{Kushner03}
H.~J. Kushner and G.~Yin.
\newblock {\em Stochastic approximation and recursive algorithms and
  applications}, volume~35.
\newblock Springer Science; Business Media, 2003.

\bibitem{lerman2014leveraging}
K.~Lerman and T.~Hogg.
\newblock Leveraging position bias to improve peer recommendation.
\newblock {\em PLoS ONE}, 9(6), 2014.

\bibitem{lim2004metaheuristics}
A.~Lim, B.~Rodrigues, and X.~Zhang.
\newblock Metaheuristics with local search techniques for retail shelf-space
  optimization.
\newblock {\em Management Science}, 50(1):117--131, 2004.

\bibitem{Ljung77}
L.~Ljung.
\newblock Analysis of recursive stochastic algorithms.
\newblock {\em IEEE Trans. Autom. Control,}, 22(4):551--575, 1977.

\bibitem{PolyaUrnModels}
H.~Mahmound.
\newblock {\em P{\'o}lya Urn Models}.
\newblock Chapman \& Hall/CRC Texts in Statistical Science, 2008.

\bibitem{Muchnik2013}
L.~Muchnik, S.~Aral, and S.~J. Taylor.
\newblock Social influence bias: A randomized experiment.
\newblock {\em Science}, 341(6146):647--651, 2013.

\bibitem{Renlund2010}
H.~Renlund.
\newblock {Generalized P{\'o}lya Urns Via Stochastic Approximation}.
\newblock {\em ArXiv e-prints 1002.3716}, Feb. 2010.

\bibitem{robinson2006algebra}
D.~J. Robinson.
\newblock {\em A course in linear algebra with applications}.
\newblock World Scientific, 2006.

\bibitem{salganik2006experimental}
M.~J. Salganik, P.~S. Dodds, and D.~J. Watts.
\newblock Experimental study of inequality and unpredictability in an
  artificial cultural market.
\newblock {\em Science}, 311(5762):854--856, 2006.

\bibitem{salganik2008leading}
M.~J. Salganik and D.~J. Watts.
\newblock Leading the herd astray: An experimental study of self-fulfilling
  prophecies in an artificial cultural market.
\newblock {\em Social Psychology Quarterly}, 71(4):338--355, 2008.

\bibitem{salganik2009web}
M.~J. Salganik and D.~J. Watts.
\newblock Web-based experiments for the study of collective social dynamics in
  cultural markets.
\newblock {\em Topics in Cognitive Science}, 1(3):439--468, 2009.

\bibitem{stumm2015}
C.~Stummer, E.~Kiesling, M.~G{\"u}nther, and R.~Vetschera.
\newblock Innovation diffusion of repeat purchase products in a competitive
  market: an agent-based simulation approach.
\newblock {\em European Journal of Operational Research}, 245(1):157--167,
  2015.

\bibitem{tucker2011does}
C.~Tucker and J.~Zhang.
\newblock How does popularity information affect choices? a field experiment.
\newblock {\em Management Science}, 57(5):828--842, 2011.

\bibitem{vandeRijt2014}
A.~van~de Rijt, S.~M. Kang, M.~Restivo, and A.~Patil.
\newblock Field experiments of success-breeds-success dynamics.
\newblock {\em Proceedings of the National Academy of Sciences},
  111(19):6934--6939, 2014.

\bibitem{ICWSM16SI}
P.~Van~Hentenryck, A.~Abeliuk, F.~Berbeglia, G.~Berbeglia, and F.~Maldonado.
\newblock Aligning popularity and quality in online cultural markets.
\newblock In {\em In the Proceedings of the International AAAI Conference on
  Web and Social Media (ICWSM 2016), Cologne, Germany, May}, 2016.

\bibitem{viglia2014please}
G.~Viglia, R.~Furlan, and A.~Ladr{\'o}n-de Guevara.
\newblock Please, talk about it! when hotel popularity boosts preferences.
\newblock {\em International Journal of Hospitality Management}, 42:155--164,
  2014.

\bibitem{wang2012}
C.~Wang, W.~Chen, and Y.~Wang.
\newblock Scalable influence maximization for independent cascade model in
  large-scale social networks.
\newblock {\em Data Mining and Knowledge Discovery}, 25(3):545--576, 2012.

\bibitem{yuan2016}
X.~Yuan and H.~B. Hwarng.
\newblock Stability and chaos in demand-based pricing under social
  interactions.
\newblock {\em European Journal of Operational Research}, 253(2):472--488,
  2016.

\end{thebibliography}

\newpage

\appendix

\section{Key Results from Bena{\"\i}m (1999) \cite{Benaim99} }
Consider the set of solutions of the differential equation \eqref{RMC}, we say that $\Upsilon=(\Upsilon_t)_{t\in\R}$ is  the {\it flow induced} by the vector field $F$, where $\Upsilon_t$ 	are the local unique solutions of \eqref{RMC} with $x^0=x_0\in \Delta^{n-1}$. Bena{\"\i}m defines the following useful concept:
A continuous function $X:\R_+\to \R^n$ is an {\it Asymptotic pseudo-trajectory} for $\Upsilon$ if for any $T>0$
	\[\lim_{t\to\infty} \sup_{0\leq h\leq T}dist(X(t+h),\Upsilon_h(X(t)))=0.\]
\noindent
Recall now that our Robbins-Monro Algorithm \eqref{RMA-MS} is defined by
\[
\phi^{k+1} = \phi^k + \gamma^{k+1} (F(\phi^k) + U^{k+1}).
\]
Let $\tau_k=\sum_{i=1}^k\gamma^i, \tau_0=0$ and define the affine interpolated process $Z(t)$:
\begin{equation}
\label{interp}
Z(t)=\phi^k+[t-\tau_k]\frac{\phi^{k+1}-\phi^k}{\gamma^{k+1}},\quad \tau_k\geq t\geq \tau_{k+1}.
\end{equation}
Consider also the map $m:\R_+\to \N$ defined by $m(t)=\sup\{k\geq0: t\geq \tau_k\}$.

\begin{proposition}[Proposition 4.1 in \cite{Benaim99}]
\label{prop:4.1}
Let $F$ be a bounded locally Lipschitz vector field. Assume that
\begin{enumerate}
\item[A1.1] For all $T>0$,
\[\lim_{l\to\infty}\sup\{\|\sum_{i=n}^{k-1}\gamma^{i+1}U^{i+1}\|:k=n+1,...,m(\tau_l+T)\}=0.\]
\item[A1.2] $\displaystyle \sup_k\|\phi^k\|<\infty.$
\end{enumerate}
Then the interpolated process $Z(t)$ is an asymptotic pseudotrajectory of the flow induced by $F$.
\end{proposition}

\begin{proposition}[Proposition 4.2 in \cite{Benaim99}]
\label{prop:4.2}
Let $\phi^k$ be the Robbins-Monro Algorithm \eqref{RMA-MS}. Suppose that, for some $q\geq 2$,
\[\sup_k \EE(\|U^{k+1}\|^q)<\infty,\]
and
\[\sum_k[\gamma^k]^{1+q/2}<\infty.\]
Then assumption $A1.1$ of Proposition A.1 holds with probability 1.
\end{proposition}

\noindent
Let $X:\R_+\to M$ be an asymptotic pseudotrajectory of an induced flow
$\Phi$, with $M$ some metric space. The limit set $L(X)$ of $X$ is the
set of limits of convergent sequences $X(t_k)$, $t_k\to\infty$.

\begin{theorem}[Theorem 5.7 $i)$ in \cite{Benaim99}]
\label{theorem:5.7}
Let $X$ be a precompact asymptotic pseudotrajectory of $\Phi$. Then $L(X)$ is Internally  Chain Transitive.
\end{theorem}

\section{Proofs}
\label{proofs}

\begin{proof}[Proof of Lemma 3.1]
The probability that item $i$ is purchased in the first step is given by
\begin{equation*}
p_i^{1st}(\phi) = \frac{v_i f(\phi_i)}{\sum\limits_{j=1}^n v_j f(\phi_j)} q_i.
\end{equation*}
The probability that item $i$ is purchased in the second step and
no item was purchased in the first step is given by
\begin{equation*}
p_i^{2nd}(\phi) =\left( \frac{\sum\limits_{j=1}^n v_j f(\phi_j)(1-q_j)}{\sum\limits_{j=1}^n v_j f(\phi_j)} \right)\frac{v_i f(\phi_i)}{\sum\limits_{j=1}^n v_j f(\phi_j)} q_i.
\end{equation*}
More generally, the probability that item $i$ is purchased in step
$m$ while  no item was purchased in earlier steps is given by
\begin{equation}
\label{noarreg}
p_i^{mth}(\phi)=\left(\frac{\sum\limits_{j=1}^n v_j f(\phi_j) (1-q_j)}{\sum\limits_{j=1}^n v_j f(\phi_j)}\right)^{m-1} \frac{v_i f(\phi_i)}{\sum\limits_{j=1}^nv_j f(\phi_j)}q_i.
\end{equation}
Let $a=(\sum\limits_{j=1}^n v_j f(\phi_j) q_j)/(\sum\limits_{j=1}^n
v_j f(\phi_j))$. Observe that, if $q_{max}=\max_{i\in\{1,...,n\}}q_i$,
then $0< a\leq q_{max}\leq 1$.  Equation \eqref{noarreg} becomes
\begin{equation*}
p_i^{mth}(\phi) =\bigg(1-a\bigg)^{m-1}  \frac{v_i f(\phi_i)}{\sum\limits_{j=1}^nv_j f(\phi_j)} q_i.
\end{equation*}
Hence the probability that the next purchase is item $i$ is given by
\begin{equation*}
p_i(\phi) =  \sum\limits_{m=0}^\infty\bigg(1-a\bigg)^{m} \frac{v_i f(\phi_i)}{\sum\limits_{j=1}^n v_j f(\phi_j)} q_i.
\end{equation*}
Since $|1-a|<1$, we use the geometric series
\begin{equation*}
\sum\limits_{m=0}^\infty\bigg(1-a\bigg)^m=\frac{1}{a},
\end{equation*}
and then, the  probability that the next purchase is item $i$ is given by
\begin{equation*}
p_i(\phi) = \frac{v_iq_if(\phi_i)}{\sum\limits_{j=1}^n v_jq_jf(\phi_j)}.
\end{equation*}

\end{proof}


\begin{proof}[Proof Theorem 4.10]

  Thanks to $H1-H2$, Proposition \ref{prop:4.2} holds for
  $q=2$. As a result, we can apply Proposition \ref{prop:4.1} and $Z(t)$ from Equation \eqref{interp}
  is almost surely an asymptotic pseudo-trajectory for the flow
  induced by $F$. As $x^t\in \Delta^{n-1}$, then $Z(t)$ is
  precompact. Finally, using Theorem \ref{theorem:5.7}, the limit set
  $L\{x^t\}_{t \geq 0}$ is an ICT for Equation \eqref{RMC}.
\end{proof}

\section{Market Shares Versus Purchases}
\label{appendix:purchases}

The condition $d_i^0=a_i$ can be relaxed and the results still hold
but the notations become more complicated. Indeed, define the variables
$\mu_i^k=\frac{ a_i+d_i^k}{\sum_j a_j+d_j^k}$, with $d_i^0=0$,
consider $\hat{a}=\sum_{j=1}^na_i$ the cumulative appeal, and $\mathbf{a}, \mathbf{d^k}$ the vectors of appeals and purchases respectively. By definition
$\sum_{j=1}^nd_j^k=k$, then we can define the probability function
$p(\mu^k)$ by
\[
p_i(\mu^k)= \frac{v_i q_i f(\mu_i^k)}{\sum_{j=1}^n v_j  q_j f(\mu_j^k)},\quad  i\in\{1,...,n\}
\]
and recover a recurrence for $\mu$ as follows:
\begin{align*}
\mu^{k+1} & =  \frac{\mathbf{a+d^k}}{\hat{a}+k +1}  + \frac{e^k}{\hat{a}+k+ 1} \\
&= \frac{\mathbf{a+d^k}}{\hat{a}+k }\frac{\hat{a}+k }{\hat{a}+k +1}  + \frac{e^k}{\hat{a}+k+ 1} \\
          & =  \mu^k\frac{\hat{a}+k }{\hat{a}+k +1}  + \frac{e^k}{\hat{a}+k+ 1} \\
          & =  \mu^k\frac{\hat{a}+k+1 }{\hat{a}+k +1} -\frac{\mu^k}{\hat{a}+k+1} + \frac{e^k}{\hat{a}+k+ 1} \\
          & =  \mu^k + \frac{1}{\hat{a}+k+1} (p(\mu^k) - \mu^k + e^k - \EE[e^k | \Fk])\\
          &= \mu^k + \hat{\gamma}^{k+1}[\hat{F}(\mu^k)+\hat{U}^{k+1}].
\end{align*}
In consequence, all the results from this paper can be translated from the $\phi$ domain to the $\mu$ domain.

\newpage
\section{Dataset}
\label{dataset}
Table \ref{qualapp} shows the values of the qualities and appeals for the independent setting (obtained from \citep{Rank}).
Table \ref{visib} shows the values of the visibilities for each position $j\in \{1,\dots, n\}$.

\begin{table}[ht]
\centering
\begin{tabular}{lllclll}
\hline
Product&	 Quality&	Appeal&		\vline&	Product &	 Quality&		Appeal\\
\hline
1&	0.8&			0.18581654&	\vline&	 26&		0.278009&	 0.35136515\\
2&	0.72&		0.28594501&	\vline&	 27&		0.2673&		0.78687609\\
3&	0.68&		0.52073051&	\vline&	 28&		0.26083&		0.7369193\\
4&	0.65&		0.81398644&	\vline&	 29&		0.2512&		0.75227893\\
5&	0.60&		0.45868017&	\vline&	 30&		0.24396&		0.32580804\\
6&	0.57&		0.15955483&	\vline&	 31&		0.23941&		0.30674759\\
7&	0.55&		0.43715743&	\vline&	 32&		0.23622&		0.91103217\\
8&	0.52005&		0.38484972&	\vline&	 33&		0.22629&		0.76236248\\
9&	0.52&		0.63739211&	\vline&	 34&		0.2214&		0.11459921\\
10&	0.4887&		0.78174105&	\vline& 	 35&		0.22013&		0.7581713\\
11&	0.48224&		0.52983037&	\vline&	 36&		0.20418&		0.76994571\\
12&	0.4586&		0.6382574&	\vline&	 37&		0.20389&		0.67408264\\
13&	0.45837&		0.80597&		\vline&	 38&		0.19535&		0.41759683\\
14&	0.432&		0.2520265&	\vline&	 39&		0.1947&		0.68898008\\
15&	0.43067&		0.37266718&	\vline&	40&		0.18248&		0.82117398\\
16&	0.38623&		0.79358615&	\vline&	41&		0.17444&		0.33890645\\
17&	0.36792&		0.19972853&	\vline&	42&		0.16867&		0.63497574\\
18&	0.35492&		0.32368825&	\vline&	43&		0.16638&		 0.16224351\\
19&	0.35374&		0.94736709&	\vline& 	44&		0.15374&		0.47778872\\
20&	0.32799&		0.50704873&	\vline&	45&		0.14542&		0.23702317\\
21&	0.32589&		0.7105828&	\vline& 	46&		0.1387&		0.49406539\\
22&	0.30411&		0.92616787&	\vline&	 47&		0.12764&		0.45956048\\
23&	0.30352&		0.64768258&	\vline&	 48&		0.12217&		0.75210134\\
24&	0.29988&		0.51815068&	\vline&	 49&		0.11418&		0.66488509\\
25&	0.2905&		0.47170285&	\vline& 	50&		0.08636&		0.80257928\\
\hline
\end{tabular}
\caption{Values of quality and appeal for the products in the independent case. Recall that the values of the appeal in the anti-correlated setting are given by $a_i=1-q_i$.}
\label{qualapp}
\end{table}

\begin{table}[ht]
\centering
\begin{tabular}{llcll}
\hline
Position & Visibility &\vline&Position & Visibility\\
\hline
1&	0.83&			\vline&	 25&				0.16583292\\
2&	0.75&			\vline&	 26&				0.15370582\\
3&	0.69&			\vline&	 27&				0.13640378\\
4&	0.62&			\vline&	28&				0.13084858\\
5&	0.58&			\vline&	 29&				0.12666812\\
6&	0.48&			\vline&	 30&				0.12429217\\
7&	0.44&			\vline&	 31&				0.12362827\\
8&	0.4&				\vline&	 32&				0.11847651\\
9&	0.37&			\vline&	 33&				0.10675012\\
10&	0.35&			\vline&	 34&				0.1001895\\
11&	0.338&			\vline& 	 35&				0.10377821\\
12&	0.321&			\vline&	 36&				0.10192779\\
13&	0.317&			\vline&	 37&				0.10484361\\
14&	0.31063943&		\vline&	 38&				0.10609265\\
15&	0.2750814&		\vline&	 39&				0.11420125\\
16&	0.25493054&		\vline&	40&				0.1260095\\
17&	0.25148059&		\vline&	41&				0.13163135\\
18&	0.23254506&		\vline&	42&				0.14843575\\
19&	0.22517471&		\vline&	43&				0.15040223\\
20&	0.22429915&		\vline& 	44&				0.15529018\\
21&	0.21502087&		\vline&	45&				0.1699023,\\
22&	0.19038769&		\vline& 	46&				0.17265442\\
23&	0.18407585&		\vline&	 47&				0.17825863\\
24&	0.18185429&		\vline&	 48&				0.18851792\\
25&	0.17013229&		\vline&	50&				0.22057129\\ 	
\hline
\end{tabular}
\caption{Values  of the visibilities for each position $j\in \{1,\dots, n\}$.}
\label{visib}
\end{table}
\end{document}